\documentclass[10pt]{article}
\usepackage{amsmath, amsthm, amssymb, amsfonts, mathrsfs, mathtools}
\usepackage{graphicx}
\usepackage{makeidx}
\usepackage[left=2cm,right=2cm,top=2cm,bottom=2cm]{geometry}
\usepackage{caption}
\usepackage{subcaption}
\usepackage[section]{placeins}
\usepackage{enumitem}
\usepackage{tikz}
\usepackage{stmaryrd}
\usepackage{textcomp}
\usepackage{lscape}
\usepackage[bottom]{footmisc}
\usepackage{makecell}

\usetikzlibrary{decorations.pathreplacing}
\tikzstyle{vertex}=[auto=left,circle,draw=black,fill=white, inner sep=1.5]

\usepackage{hyperref}
\hypersetup{colorlinks=true, citecolor={blue}, linkcolor=blue, filecolor=magenta, urlcolor=cyan}

\newtheorem{theorem}{Theorem}[section]

\newtheorem{lema}[theorem]{Lemma}

\title{Perfect state transfer on Cayley graphs over a non-abelian group of order $8n$}

\author{ Akash Kalita and Bikash Bhattacharjya\\
Department of Mathematics\\
Indian Institute of Technology Guwahati, India\\
akash.kalita@iitg.ac.in, b.bikash@iitg.ac.in }

\date{}
\begin{document}
\raggedbottom
\maketitle

\vspace{-0.3in}

\begin{center}{\textbf{Abstract}}\end{center}
The \textit{transition matrix} of a graph $\Gamma$ with adjacency matrix $A$ is defined by $H(\tau ) := \exp(-\mathbf{i}\tau A)$, where $\tau  \in \mathbb{R}$ and $\mathbf{i} = \sqrt{-1}$. The graph $\Gamma$ exhibits \textit{perfect state transfer} (PST) between the vertices $u$ and $v$ if there exists $\tau_0(>0)\in \mathbb{R}$ such that $\lvert H(\tau_0)_{uv} \rvert = 1$. For a positive integer $n$, the group $V_{8n}$ is defined as $V_{8n} := \langle a,b \colon a^{2n} = b^{4} = 1, ba = a^{-1}b^{-1}, b^{-1}a = a^{-1}b \rangle$. In this paper, we study the existence of perfect state transfer on Cayley graphs $\text{Cay}(V_{8n}, S)$. We present some necessary and sufficient conditions for the existence of perfect state transfer on $\text{Cay}(V_{8n}, S)$.

\vspace*{0.3cm}

\section{Introduction}

The concept of quantum walks on graphs was introduced by Farhi and Gutmann~\cite{Farhi} in the year 1998. Quantum walks on finite graphs provide useful  models for quantum transport phenomena discovered by Bose~\cite{Bose} in 2003. Quantum walks are important tools in quantum computation and information theory, and it can be used to describe the fidelity of information transfer in a network of interacting qubits. Christandl et al.~\cite{Christandl} proposed a class of qubit networks that exhibit perfect state transfer. Perfect state transfer has great importance due to its applications in quantum information processing, quantum communication networks and cryptography.

Let $\Gamma$ be a finite simple connected graph with adjacency matrix $A$. The \textit{transition matrix} of $\Gamma$ is defined by
\begin{align*}
H(\tau ) :=  \exp(-\mathbf{i}\tau A) = \sum_{s=0}^{\infty} \frac{(-\mathbf{i}\tau A)^s}{s!}, \text{ where } \tau  \in \mathbb{R} \text{ and } \mathbf{i} = \sqrt{-1}. 
\end{align*}

A graph $\Gamma$ exhibits \textit{perfect state transfer} (PST) between the vertices $u$ and $v$ at time $\tau_0(>0)\in \mathbb{R}$ if  the $uv$-th entry of $H(\tau_0)$ has absolute value one. We say that $\Gamma$ exhibits \textit{periodicity} at the vertex $u$ at time $\tau_0$ if the $uu$-th entry of $H(\tau_0)$ has absolute value one. The graph $\Gamma$ is considered to exhibit \textit{periodicity}
if it exhibits the property of periodicity across all its vertices simultaneously at the same time. 

Let $G$ be a finite group. Let  $S$ be a symmetric subset of $G \setminus\{1\}$, that is, $S  = \{s^{-1}: s\in S\}$. The \textit{Cayley} graph of $G$ with respect to $S$, denoted $\text{Cay}(G, S)$, is a graph whose vertices are the elements of $G$ and there exists an edge between distinct vertices $g,h\in G$ if $gh^{-1}\in S$. The set $S$ is called the \textit{connection} set of $\text{Cay}(G, S)$. If $Sg = gS$ for all $g \in G$, then  $\text{Cay}(G, S)$ is called a \textit{normal} Cayley graph. Since $S$ is symmetric, $\text{Cay}(G, S)$ is a simple graph. We assume $G = \langle S\rangle $ to ensure that $\text{Cay}(G, S)$ is connected. The adjacency matrix of $\text{Cay}(G, S)$ is  $A  := [a_{gh}]_{g,h\in G}$, where
\[a_{gh} = \left\{ \begin{array}{cr}
     1 & \mbox{ if }
     gh^{-1} \in S \\
     0 &\textnormal{ otherwise. }
\end{array}\right.\]

Perfect state transfer has been studied for several families of graphs in the last two decades. Angeles-Canul et al.~\cite{Canul} investigated PST in integral circulant graphs and the join of graphs. Coutinho and Godsil~\cite{Coutinho 1} explored PST in products and covers of graphs. Pal and Bhattacharjya~\cite{Pal} studied PST on NEPS of the path on three vertices. Godsil~\cite{Godsil 2} offered a survey on PST and related questions up to 2011. He~\cite{Godsil 1} explained the close relationship between the existence of PST on certain graphs and association schemes. Notably, in~\cite{Godsil 3}  Godsil presented a complete characterization of PST on simple connected graphs in terms of its eigenvalues.  

Cayley graphs are good candidates for exhibiting PST due to their nice algebraic structure. Basic et al.~\cite{Basic 3, Basic 2, Basic 1, Basic 4} , Cheung and Godsil~\cite{Cheung} and Bernasconi et al.~\cite{Bernasconi} presented some criterion on circulant graphs and cubelike graphs exhibiting PST. Tan et al.~\cite{Tan} presented a characterization on connected abelian Cayley graphs exhibiting PST. They showed that many of the previous results on periodicity and existence of PST in circulant graphs and cubelike graphs can be derived in unified and more simple ways. However, relatively little information are known about Cayley graphs over non-abelian groups exhibiting PST. Cao and Feng~\cite{Cao 1} investigated PST on Cayley graphs over dihedral groups. Subsequently, Arezoomand et al.~\cite{Arezoomand} and Luo et al.~\cite{Cao 2} explored PST on Cayley graphs over dicyclic groups and semi-dihedral groups, respectively. Recently, Khalilipour and Ghorbani~\cite{Khalilipour} studied PST on Cayley graphs over the group $U_{6n}$, where $U_{6n} := \langle a, b \colon a^{2n} = b^{3} = 1, a^{-1}ba = b^{-1} \rangle$ and $n$ is a positive integer.

The group $V_{8n}$ is defined by $V_{8n} := \langle a,b \colon a^{2n} = b^{4} = 1, ba = a^{-1}b^{-1}, b^{-1}a = a^{-1}b \rangle$, where $n$ is a positive integer. It is clear that $V_{8n}$ is a non-abelian group of order $8n$. Cheng et al.~\cite{Tao Cheng} studied integral Cayley graphs over $V_{8n}$ for odd $n$. They gave some necessary and sufficient conditions for the integrality of $\text{Cay}(V_{8p}, S)$, where $p$ is an odd prime. In this paper, we consider the existence of PST on Cayley graphs over $V_{8n}$. Using the irreducible representations of $V_{8n}$, some necessary and  sufficient conditions for a normal $\text{Cay}(V_{8n}, S)$ exhibiting PST are obtained. 

The rest of the paper is organized as follows. In Section~\ref{sec2}, we give the description of the irreducible representations of  $V_{8n}$ and eigenvectors of normal $\text{Cay}(V_{8n}, S)$. The existence of PST on normal $\text{Cay}(V_{8n}, S)$ is explored in Section~\ref{sec4}.

\section{Representations of $V_{8n}$ and the eigenvectors of $\text{Cay}(V_{8n}, S)$}\label{sec2}
A \textit{representation} of a finite group $G$ is a homomorphism $\theta \colon G \rightarrow \mathrm{GL}(U)$, where $\mathrm{GL}(U)$ is the group of all  automorphisms of a finite-dimensional and non-zero complex vector space $U$. The dimension of $U$ is called the \textit{degree} of $\theta$. Two representations $\theta$ and $\psi$ of $G$ on $U$ and $W$, respectively, are \textit{equivalent} if there is an isomorphism
$T \colon U \rightarrow W$ such that $\psi(g) = T\theta(g)T^{-1}$ for all $g\in G$. 

Let $\theta \colon G \rightarrow \mathrm{GL}(U)$ be a representation of a group $G$. The \textit{character} $\chi_\theta \colon G \rightarrow \mathbb{C}$
of $\theta$ is defined by setting $\chi_\theta(g) = \mathrm{Tr}(\theta(g))$ for all $g \in G$, where $\mathrm{Tr}(\theta(g))$ is the trace of $\theta(g)$. A subspace $W$ of $U$ is said to be $G$-\textit{invariant} if $\theta(g)w\in W$ for all $g\in G$ and $w\in W$. Obviously, $\{\textbf{0}\}$ and $U$ are $G$-invariant subspaces of $U$, called the \textit{trivial} subspaces. If $U$ has no non-trivial $G$-invariant subspaces, then $\theta$ is called an \textit{irreducible representation} and $\chi_\theta$ an \textit{irreducible character} of $G$. 
 
Note that the group $V_{8n}$ can be written as $V_{8n} = \{a^{r},a^{r}b,a^{r}b^{2},a^{r}b^{3} \colon 0\leq r\leq {2n-1}\}$. For odd values of $n$, the group $V_{8n}$ has the following $2n+3$ conjugacy classes.
\begin{align*}
&\{1\},~\{b^{2}\},~\{a^{j}b^{k} \colon  j ~\mathrm{even},~ k = 1~ \mathrm{or}~k = 3 \},~\{a^{j} b^{k} \colon j ~ \mathrm{odd},~ k = 1  ~ \mathrm{or}  ~ k = 3\},\\
&\{a^{2r+1},a^{-2r-1}b^{2}\}~\text{ for }~r \in \{0, \ldots, {n - 1}\}~\text{ and}~\{a^{2s},a^{-2s}\},\{a^{2s}b^{2},a^{-2s}b^{2}\}~\text{ for } ~s \in \{1, \ldots, {(n-1)/2}\}.
\end{align*}
For even values of $n$, the group $V_{8n}$ has the following $2n+6$ conjugacy classes.
\begin{align*}
&\{1\}, ~\{b^2\}, ~\{a^n\}, ~\{a^nb^2\},~\{a^{2k}b^{(-1)^k} \colon 0 \leq k \leq {n - 1}\},~
\{a^{2k}b^{(-1)^{k + 1}} \colon 0 \leq k \leq {n - 1}\},\\
&\{a^{2k + 1}b^{(-1)^k} \colon 0 \leq k \leq {n - 1}\},~
\{a^{2k + 1}b^{(-1)^{k + 1}} \colon 0 \leq k \leq {n - 1}\},\\
&\{a^{2r+1},a^{-{(2r + 1)}}b^{2}\}~\text{ for }~r \in \{0,\ldots,{n - 1}\}\text{ and }
\{a^{2s}, a^{-2s}\}, \{a^{2s}b^2, a^{-2s}b^2\}~\text{ for }~s \in \{1,\ldots,{n/2 - 1}\}.
\end{align*}

The conjugacy classes of $V_{8n}$ are used to write its character table. Throughout the paper, we consider $\omega =\exp(\frac{\pi \textbf{i}}{n})$, a primitive $2n$-th root of unity. The following lemma presents the irreducible representations and the irreducible characters of $V_{8n}$.
\begin{lema}\label{Lemma2.1}
\cite{Darafsheh,James 2001} 
Let $n$ be a positive integer. Then
\begin{enumerate}
\item[(i)] The irreducible representations of $V_{8n}$ are listed in Table 1 for $n$ odd and in Table 2 for $n$ even.
\item[(ii)] The character table of $V_{8n}$ is listed in Table 3 for $n$ odd, in Table 4 for $n \equiv 0
~(mod~ 4)$ and in Table 5 for $n \equiv 2~(mod ~4)$. 
\end{enumerate}
\end{lema}
The following lemma determines the eigenvalues and the corresponding eigenvectors of the adjacency matrix of a normal Cayley graph.

\begin{lema}\label{Lemma2.2}\cite{Steinberg 2012} 
Let $G = \{g_1,\ldots,g_n\}$ be a finite group and $\phi^{(1)},\ldots,{\phi}^{(m)}$ be  a complete set of unitary representatives of the equivalence classes of irreducible representations of $G$. Let $\chi_k$ and $d_k$ be the character and the degree of ${\phi}^{(k)}$, respectively for $1 \leq k \leq m$. Let $S$ be a conjugation-closed symmetric subset of $G$. Then the eigenvalues of $\mathrm{Cay}(G,S)$ are ${\lambda}_1,\ldots,{\lambda}_m$, where
\begin{align*}
   \lambda_k =  \frac{1}{d_k}\sum_{s\in S}{\chi_k(s)} ~\text{for}~  1\leq k \leq m,
\end{align*}
and $\lambda_k$ has multiplicity $d_k^{2}$ for each $k \in \{1,\ldots,m\}$. Moreover, the vectors
\begin{align*}
 \textbf{v}_{ij}^{(k)} = \sqrt{\frac{d_k}{|G|}} \left[{\phi}_{ij}^{(k)}(g_1),\ldots,{\phi}_{ij}^{(k)}(g_n)\right]^t \text{ for }
1 \leq i,j \leq d_k                
\end{align*}
form an orthonormal basis for the eigenspace corresponding to the eigenvalue $\lambda_k$ for each $k \in \{1,\ldots,m\}$.
\end{lema}

\begin{table}
\caption{Complete set of irreducible unitary inequivalent representations of $V_{8n}$ for $n$ odd}
\centering
\begin{tabular}{|c| c| c|}
\hline
& $a$ & $b$\\
\hline
${\theta}_1$ & $1$ & $1$\\
\hline
${\theta}_2$ & $1$ & $-1$\\
\hline
${\theta}_3$ & $-1$ & $1$\\
\hline
${\theta}_4$ & $-1$ & $-1$\\
\hline
${\psi}_j$ $(0 \leq j \leq n-1)$

& $ \begin{pmatrix}
     {\omega}^{2j} & 0 \\
     0 & -{\omega}^{-2j}
    \end{pmatrix}$

& $ \begin{pmatrix}0&1\\-1&0\end{pmatrix}$\\\hline
${\phi}_k$ $(1 \leq k \leq n-1)$
& $ \begin{pmatrix}
     {\omega}^{k} & 0 \\
     0 & {\omega}^{-k}
\end{pmatrix}$              &     $ \begin{pmatrix}
                                   0 & 1 \\
                                   1 & 0
                                   \end{pmatrix}$  \\

\hline

\end{tabular}
\end{table}
\begin{table}
\caption{Complete set of irreducible unitary inequivalent representations of $V_{8n}$ for $n$ even}
\centering
\begin{tabular}{|c| c| c|}
\hline
& $a$ & $b$\\
\hline
${\theta}_1$ & $1$ & $1$\\
\hline
${\theta}_2$ & $\textbf{i}$ & $-\textbf{i}$\\
\hline
${\theta}_3$ & $-1$ & $-1$\\
\hline
${\theta}_4$ & $-\textbf{i}$ & $\textbf{i}$\\
\hline
${\theta}_5$ & $1$ & $-1$\\
\hline
${\theta}_6$ & $\textbf{i}$ & $\textbf{i}$\\
\hline
${\theta}_7$ & $-1$ & $1$\\
\hline
${\theta}_8$ & $-\textbf{i}$ & $-\textbf{i}$\\
\hline
${\psi}_j$ $(1 \leq j \leq n-1)$

& $ \begin{pmatrix}
     \omega^{j} & 0 \\
     0 & \omega^{-j}
    \end{pmatrix}$

& $ \begin{pmatrix}0 & \textbf{i}\\-\textbf{i}&0\end{pmatrix}$\\\hline
$\phi_k$ $(1 \leq k \leq n-1)$
& $ \begin{pmatrix}
     \textbf{i}\omega^{k} & 0 \\
     0 & \textbf{i}\omega^{-k}
\end{pmatrix}$              &     $ \begin{pmatrix}
                                 0 & 1 \\
                                  -1 & 0
                                   \end{pmatrix}$  \\

\hline

\end{tabular}
\end{table}

\begin{table}

\caption{Character table of $V_{8n}$ for $n$ odd}
\begin{tabular}{|c| c |c| c| c| c| c |c|}

\hline
  &$1$ & $b^{2}$ & \thead{$a^{2r+1}$\\ $(0 \leq r \leq n-1)$} & $a^{2s}$ & \thead{$a^{2s}b^{2}$\\ $(1 \leq s \leq (n-1)/2)$} & $b$ & $ab$  \\ 
\hline
${\chi}_1$ & $1$ & $1$ & $1$ & $1$ & $1$ & $1$ & $1$ \\
\hline
${\chi}_2$ & $1$ & $1$ & $1$ & $1$ & $1$ & $-1$ & $-1$ \\
\hline
${\chi}_3$ & $1$ & $1$ & $-1$ & $1$ & $1$ & $1$ & $-1$ \\
\hline
${\chi}_4$ & $1$ & $1$ & $-1$ & $1$ & $1$ & $-1$ & $1$   \\                                     
\hline
${\zeta}_j$ $(0\leq j\leq n-1)$ & $2$ & $-2$ & ${\omega}^{2j(2r+1)}-{\omega}^{-2j(2r+1)}$ & ${\omega}^{4js} + {\omega}^{-4js}$ & $-{\omega}^{4js}-{\omega}^{-4js}$ & $0$ & $0$ \\
\hline
${\xi}_k$ $(1\leq k \leq n-1)$ & $2$ & $2$ & ${\omega}^{k(2r+1)}+{\omega}^{-k(2r+1)}$ & ${\omega}^{2ks}+{\omega}^{-2ks}$ & ${\omega}^{2ks}+{\omega}^{-2ks}$ & $0$ & $0$  \\
\hline

\end{tabular}
\end{table}

\begin{landscape}
\begin{table}
\caption{Character table of $V_{8n}$ for $n \equiv 0~(\mathrm{mod}~4)$}
\begin{tabular}{|c|c|c|c|c|c|c|c|c|c|c|c|c|c|c|}
\hline
 & $1$ & $b^2$ & $a^n$ & $a^nb^2$ & $a^{4m+1}$ & $a^{4m+3}$ & $a^{4s}$ & $a^{4p+2}$ & $a^{4s}b^2$ & $a^{4p+2}b^2$ & $b$ & $b^{-1}$ & $ab$ & $ab^{-1}$\\ 
\hline
${\chi}_1$ & $1$ & $1$ & $1$ & $1$ & $1$ & $1$ & $1$ & $1$ & $1$ & $1$ & $1$ & $1$ & $1$ & $1$\\
\hline
${\chi}_2$ & $1$ & $-1$ & $-1$ & $1$ & $\textbf{i}$  & $-\textbf{i}$ & $1$ & $-1$ & $-1$ & $1$ & $-\textbf{i}$ & $\textbf{i}$ & $1$ & $-1$\\
\hline
${\chi}_3$ & $1$ & $1$ & $1$ & $1$ & $-1$ & $-1$ & $1$ & $1$ & $1$ & $1$ & $-1$ & $-1$ & $1$ & $1$\\
\hline
${\chi}_4$ & $1$ & $-1$ & $-1$ & $1$ & $-\textbf{i}$ & $\textbf{i}$ & $1$ & $-1$ & $-1$ & $1$ & $\textbf{i}$ & $-\textbf{i}$ & $1$ & $-1$\\                                     
\hline
${\chi}_5$ & $1$ & $1$ & $1$ & $1$ & $1$ & $1$ & $1$ & $1$ & $1$ & $1$ & $-1$ & $-1$ & $-1$ & $-1$\\
\hline
${\chi}_6$ & $1$ & $-1$ & $-1$ & $1$ & $\textbf{i}$ & $-\textbf{i}$ & $1$ & $-1$ & $-1$ & $1$ & $\textbf{i}$ & $-\textbf{i}$ & $-1$ & $1$\\
\hline
${\chi}_7$ & $1$ & $1$ & $1$ & $1$ & $-1$ & $-1$ & $1$ & $1$ & $1$ & $1$ & $1$ & $1$ & $-1$ & $-1$\\
\hline
${\chi}_8$ & $1$ & $-1$ & $-1$ & $1$ & $-\textbf{i}$ & $\textbf{i}$ & $1$ & $-1$ & $-1$ & $1$ & $-\textbf{i}$ & $\textbf{i}$ & $-1$ & $1$\\                                     
\hline
${\zeta}_j$ & $2$ & $2$ & $2{(-1)}^j$ & $2{(-1)}^j$ & $\alpha^{j(4m+1)}$ & $\alpha^{j(4m+3)}$ & $\alpha^{j(4s)}$ & $\alpha^{j(4p+2)}$ & $\alpha^{j(4s)}$ & $\alpha^{j(4p+2)}$ & $0$ & $0$ & $0$ & $0$\\
\hline
${\xi}_k$ & $2$ & $-2$ & $2{(-1)}^k$ & $-2{(-1)}^k$ & $\textbf{i}{\alpha}^{j(4m+1)}$ & $-\textbf{i}{\alpha}^{j(4m+3)}$  & $\alpha^{j(4s)}$ & $-{\alpha}^{j(4p+2)}$ & $-\alpha^{j(4s)}$ & $\alpha^{j(4p+2)}$ & $0$ & $0$ & $0$ & $0$\\
\hline

\end{tabular}
\end{table}
Here $\alpha^{jr} = \omega^{jr} + \omega^{-jr} = 2\cos(\frac{{\pi}jr}{n})$,
$\alpha^{kr} = \omega^{kr} + \omega^{-kr} = 2\cos(\frac{{\pi}kr}{n})$,

$m \in \{0,\ldots,{n/2 - 1}\}$,
$s \in \{1,\ldots,{n/4 - 1}\}$,
$p \in \{0,\ldots,{n/4 - 1}\}$, and 
$j, k \in \{1,\ldots,{n - 1}\}$.

\end{landscape}

\begin{landscape}
\begin{table}
\caption{Character table of $V_{8n}$ for $n \equiv 2~(\mathrm{mod}~4)$}
\begin{tabular}{|c|c|c|c|c|c|c|c|c|c|c|c|c|c|c|}
\hline
 & $1$ & $b^2$ & $a^n$ & $a^nb^2$ & $a^{4m+1}$ & $a^{4m+3}$ & $a^{4s}$ & $a^{4p+2}$ & $a^{4s}b^2$ & $a^{4p+2}b^2$ & $b$ & $b^{-1}$ & $ab$ & $ab^{-1}$\\ 
\hline
${\chi}_1$ & $1$ & $1$ & $1$ & $1$ & $1$ & $1$ & $1$ & $1$ & $1$ & $1$ & $1$ & $1$ & $1$ & $1$\\
\hline
${\chi}_2$ & $1$ & $-1$ & $1$ & $-1$ & $\textbf{i}$  & $-\textbf{i}$ & $1$ & $-1$ & $-1$ & $1$ & $-\textbf{i}$ & $\textbf{i}$ & $1$ & $-1$\\
\hline
${\chi}_3$ & $1$ & $1$ & $1$ & $1$ & $-1$ & $-1$ & $1$ & $1$ & $1$ & $1$ & $-1$ & $-1$ & $1$ & $1$\\
\hline
${\chi}_4$ & $1$ & $-1$ & $1$ & $-1$ & $-\textbf{i}$ & $\textbf{i}$ & $1$ & $-1$ & $-1$ & $1$ & $\textbf{i}$ & $-\textbf{i}$ & $1$ & $-1$\\                                     
\hline
${\chi}_5$ & $1$ & $1$ & $1$ & $1$ & $1$ & $1$ & $1$ & $1$ & $1$ & $1$ & $-1$ & $-1$ & $-1$ & $-1$\\
\hline
${\chi}_6$ & $1$ & $-1$ & $1$ & $-1$ & $\textbf{i}$ & $-\textbf{i}$ & $1$ & $-1$ & $-1$ & $1$ & $\textbf{i}$ & $-\textbf{i}$ & $-1$ & $1$\\
\hline
${\chi}_7$ & $1$ & $1$ & $1$ & $1$ & $-1$ & $-1$ & $1$ & $1$ & $1$ & $1$ & $1$ & $1$ & $-1$ & $-1$\\
\hline
${\chi}_8$ & $1$ & $-1$ & $1$ & $-1$ & $-\textbf{i}$ & $\textbf{i}$ & $1$ & $-1$ & $-1$ & $1$ & $-\textbf{i}$ & $\textbf{i}$ & $-1$ & $1$\\                                     
\hline
${\zeta}_j$ & $2$ & $2$ & $2{(-1)}^j$ & $2{(-1)}^j$ & $\alpha^{j(4m+1)}$ & $\alpha^{j(4m+3)}$ & $\alpha^{j(4s)}$ & $\alpha^{j(4p+2)}$ & $\alpha^{j(4s)}$ & $\alpha^{j(4p+2)}$ & $0$ & $0$ & $0$ & $0$\\
\hline
${\xi}_k$ & $2$ & $-2$ & $-2{(-1)}^k$ & $2{(-1)}^k$ & $\textbf{i}{\alpha}^{j(4m+1)}$ & $-\textbf{i}{\alpha}^{j(4m+3)}$  & $\alpha^{j(4s)}$ & $-{\alpha}^{j(4p+2)}$ & $-\alpha^{j(4s)}$ & $\alpha^{j(4p+2)}$ & $0$ & $0$ & $0$ & $0$\\
\hline

\end{tabular}
\end{table}
Here $\alpha^{jr} = \omega^{jr} + \omega^{-jr} = 2\cos(\frac{{\pi}jr}{n})$, $\alpha^{kr} = \omega^{kr} + \omega^{-kr} = 2\cos(\frac{{\pi}kr}{n})$,

$m \in \{0,\ldots,{n/2 - 1}\}$,
$s \in \{1,\ldots,{n/4 - 1}\}$,
$p \in \{0,\ldots,{n/4 - 1}\}$, and
$j, k \in \{1,\ldots,{n - 1}\}$.

\end{landscape}
In what follows, we consider only the normal Cayley graphs $\text{Cay}(V_{8n},S)$. We also consider the ordering $1,a,a^{2},\ldots,a^{2n-1},b,ab,\ldots,a^{2n-1}b,b^{2},ab^{2},\ldots,a^{2n-1}b^{2},b^{3},ab^{3},\ldots,a^{2n-1}b^{3}$ for the elements of $V_{8n}$. Recall that the adjacency matrix of $\text{Cay}(V_{8n},S)$
has an eigenvalue corresponding to each irreducible inequivalent representations of $V_{8n}$. In the following two subsections, we use Lemma~\ref{Lemma2.1} and Lemma~\ref{Lemma2.2} to compute the eigenvectors of $\text{Cay}(V_{8n},S)$.


\subsection{Eigenvectors of normal $\text{Cay}(V_{8n}, S)$ for odd $n$}\label{subsection 2.1}
Let $n$ be an odd positive integer. The adjacency matrix of  $\text{Cay}(V_{8n},S)$ has the following eigenvectors corresponding to the one-dimensional irreducible representations $\theta_1,\theta_2,\theta_3$ and $\theta_4$, respectively.
\begin{align*}
\textbf{u}_1 = & \frac{1}{\sqrt{8n}} [1,\ldots,1]^t ,\\
\textbf{u}_2 = & \frac{1}{\sqrt{8n}} [1,\ldots,1,~~-1,\ldots,-1,~~1,\ldots,1,~~-1,\ldots,-1]^t ,\\
\textbf{u}_3 =& \frac{1}{\sqrt{8n}}[1,-1,\ldots,1,-1,~~1,-1,\ldots,1,-1,~~1,-1,\ldots,1,-1,~~1,-1,\ldots,1,-1]^t ~\mathrm{and}~\\
\textbf{u}_4 =&  \frac{1}{\sqrt{8n}}[1,-1,\ldots,1,-1,~~-1,1,\ldots,-1,1,~~1,-1,\ldots,1,-1,~~-1,1,\ldots,-1,1]^t .
\end{align*}
The adjacency matrix of $\text{Cay}(V_{8n},S)$ has the following eigenvectors corresponding to the two-dimensional irreducible representations
${\psi}_j$ for $0 \leq j \leq n-1$.
\begin{align*}
{\textbf{u}_j}^{(1)} &=\frac{1}{\sqrt{4n}} [\{{\omega}^{2rj}\}_{r=0}^{2n-1},~\textbf{0},~{\{-\omega}^{2rj}\}_{r=0}^{2n-1},~\textbf{0}]^t,\\
{\textbf{u}_j}^{(2)} &=\frac{1}{\sqrt{4n}} [\textbf{0},~\{{\omega}^{2rj}\}_{r=0}^{2n-1},~\textbf{0},~\{{-\omega}^{2rj}\}_{r=0}^{2n-1}]^t,\\
{\textbf{u}_j}^{(3)} &=\frac{1}{\sqrt{4n}} [\textbf{0},~\{-(-{\omega}^{-2j})^r\}_{r=0}^{2n-1},~\textbf{0},~\{(-{\omega}^{-2j})^r\}_{r=0}^{2n-1}]^t~~\mathrm{and}~\\ 
{\textbf{u}_j}^{(4)} &=\frac{1}{\sqrt{4n}} [\{(-{\omega}^{2j})^r\}_{r=0}^{2n-1},~\textbf{0},~\{-(-{\omega}^{2j})^r\}_{r=0}^{2n-1},~\textbf{0}]^t.
\end{align*}

Further, the adjacency matrix of $\text{Cay}(V_{8n},S)$ has the following eigenvectors corresponding to the two-dimensional irreducible representations
$\phi_k$ for $1 \leq k \leq n-1$.
\begin{align*}
{\textbf{v}_k}^{(1)} &=\frac{1}{\sqrt{4n}} [\{{\omega}^{rk}\}_{r=0}^{2n-1},~\textbf{0},~\{{\omega}^{rk}\}_{r=0}^{2n-1},~\textbf{0}]^t,\\
{\textbf{v}_k}^{(2)} &=\frac{1}{\sqrt{4n}}[\textbf{0},~\{{\omega}^{rk}\}_{r=0}^{2n-1},~\textbf{0},~\{{\omega}^{rk}\}_{r=0}^{2n-1}]^t,\\
{\textbf{v}_k}^{(3)} &=\frac{1}{\sqrt{4n}} [\textbf{0},~\{{\omega}^{-rk}\}_{r=0}^{2n-1},~\textbf{0},~\{{\omega}^{-rk}\}_{r=0}^{2n-1}]^t~~\mathrm{and}~\\
{\textbf{v}_k}^{(4)} &=\frac{1}{\sqrt{4n}} [\{{\omega}^{-rk}\}_{r=0}^{2n-1},~\textbf{0},~\{{\omega}^{-rk}\}_{r=0}^{2n-1},~\textbf{0}]^t.
\end{align*}
Here $\textbf{0}$ denotes the row of length $2n$ consisting of all the zeros. Also, for complex numbers $z$ and $\varepsilon$, $\{{\varepsilon}z^r\}_{r = 0}^{2n - 1}$ denotes the vector $[{\varepsilon}, {\varepsilon}z, \ldots, {\varepsilon}z^{2n-1}]^t$. Note that the eigenvectors mentioned in this subsection form an orthonormal basis of ${\mathbb{C}}^{8n}$.
\subsection{Eigenvectors of normal $\text{Cay}(V_{8n}, S)$ for even $n$}\label{subsection 2.2}
Let $n$ be an even positive integer. The adjacency matrix of $\text{Cay}(V_{8n},S)$ has the following eigenvectors corresponding to the one-dimensional irreducible representations $\theta_1,\ldots,\theta_8$, respectively.
\begin{align*}
\textbf{u}_1 &= \frac{1}{\sqrt{8n}}[1, \ldots, 1]^t,\\
\textbf{u}_2 &= \frac{1}{\sqrt{8n}}[\{\textbf{i}^r\}_{r = 0}^{2n - 1},~\{\textbf{i}^{r+3}\}_{r = 0}^{2n - 1}, ~\{-{\textbf{i}}^r\}_{r = 0}^{2n - 1},~ \{{\textbf{i}}^{r+1}\}_{r = 0}^{2n - 1}]^t,\\
\textbf{u}_3 &=  \frac{1}{\sqrt{8n}}[1,-1,\ldots,1,-1,~~-1,1,\ldots,-1,1,~~1,-1,\ldots,1,-1,~~-1,1,\ldots,-1,1]^t,\\
\textbf{u}_4 &= \frac{1}{\sqrt{8n}}[\{({-\textbf{i}})^r\}_{r = 0}^{2n - 1},~ \{({-\textbf{i}})^{r+3}\}_{r = 0}^{2n - 1}, ~\{({-\textbf{i}})^{r+2}\}_{r = 0}^{2n - 1},~ \{({-\textbf{i}})^{r+1}\}_{r = 0}^{2n - 1}]^t,\\
\textbf{u}_5 &=  \frac{1}{\sqrt{8n}} [1,\ldots,1,~~-1,\ldots,-1,~~1,\ldots,1,~~-1,\ldots,-1]^t ,\\
\textbf{u}_6 &= \frac{1}{\sqrt{8n}}[\{\textbf{i}^r\}_{r = 0}^{2n - 1}, ~\{\textbf{i}^{r+1}\}_{r = 0}^{2n - 1},~ \{{\textbf{i}}^{r+2}\}_{r = 0}^{2n - 1}, ~\{{\textbf{i}}^{r+3}\}_{r = 0}^{2n - 1}]^t,\\
\textbf{u}_7 &= \frac{1}{\sqrt{8n}}[1,-1,\ldots,1,-1,~~1,-1,\ldots,1,-1,~~1,-1,\ldots,1,-1,~~1,-1,\ldots,1,-1]^t ~\mathrm{and}~\\
\textbf{u}_8 &= \frac{1}{\sqrt{8n}}[\{({-\textbf{i}})^r\}_{r = 0}^{2n - 1},~ \{({-\textbf{i}})^{r+1}\}_{r = 0}^{2n - 1}, ~\{({-\textbf{i}})^{r+2}\}_{r = 0}^{2n - 1}, ~\{({-\textbf{i}})^{r+3}\}_{r = 0}^{2n - 1}]^t.
\end{align*}
The adjacency matrix of $\text{Cay}(V_{8n},S)$ has the following eigenvectors corresponding to the two-dimensional irreducible representations
${\psi}_j$ for $1 \leq j \leq n-1$.
\begin{align*}
{\textbf{u}_j}^{(1)} &=\frac{1}{\sqrt{4n}} [\{{\omega}^{rj}\}_{r=0}^{2n-1},~\textbf{0},~{\{\omega}^{rj}\}_{r=0}^{2n-1},~\textbf{0}]^t,\\
{\textbf{u}_j}^{(2)} &=\frac{1}{\sqrt{4n}} [\textbf{0},~\{\textbf{i}\omega^{rj}\}_{r=0}^{2n-1},~\textbf{0},~\{\textbf{i}\omega^{rj}\}_{r=0}^{2n-1}]^t,\\
{\textbf{u}_j}^{(3)} &=\frac{1}{\sqrt{4n}} [\textbf{0},~\{-\textbf{i}\omega^{-rj}\}_{r=0}^{2n-1},~\textbf{0},~\{-\textbf{i}\omega^{-rj}\}_{r=0}^{2n-1}]^t~~\mathrm{and}~\\ 
{\textbf{u}_j}^{(4)} &=\frac{1}{\sqrt{4n}} [\{{\omega}^{-rj}\}_{r=0}^{2n-1},~\textbf{0},~{\{\omega}^{-rj}\}_{r=0}^{2n-1},~\textbf{0}]^t.
\end{align*}

Further, the adjacency matrix of $\text{Cay}(V_{8n},S)$ has the following eigenvectors corresponding to the two-dimensional irreducible representations
$\phi_k$ for $1 \leq k \leq n-1$.
\begin{align*}
{\textbf{v}_k}^{(1)} &=\frac{1}{\sqrt{4n}} [\{(\textbf{i}{\omega}^{k})^r\}_{r=0}^{2n-1},~\textbf{0},~\{-(\textbf{i}{\omega}^{k})^r\}_{r=0}^{2n-1},~\textbf{0}]^t,\\
{\textbf{v}_k}^{(2)} &=\frac{1}{\sqrt{4n}}[\textbf{0},~\{(\textbf{i}{\omega}^{k})^r\}_{r=0}^{2n-1},~\textbf{0},~\{-(\textbf{i}{\omega}^{k})^r\}_{r=0}^{2n-1}]^t,\\
{\textbf{v}_k}^{(3)} &=\frac{1}{\sqrt{4n}} [\textbf{0},~\{-(\textbf{i}{\omega}^{-k})^r\}_{r=0}^{2n-1},~\textbf{0},~\{(\textbf{i}{\omega}^{-k})^r\}_{r=0}^{2n-1}]^t~~\mathrm{and}~\\
{\textbf{v}_k}^{(4)} &=\frac{1}{\sqrt{4n}} [\{(\textbf{i}{\omega}^{-k})^r\}_{r=0}^{2n-1},~\textbf{0},~\{-(\textbf{i}{\omega}^{-k})^r\}_{r=0}^{2n-1},~\textbf{0}]^t.
\end{align*}
These eigenvectors form an orthonormal basis of ${\mathbb{C}}^{8n}$.
\section{PST on Cayley graphs}\label{sec4}
Let $\Gamma$ be a simple graph on $n$ vertices, and let $A$ be its adjacency matrix. We denote by $\mathrm{Spec}(\Gamma)$ the set consisting all the eigenvalues of $\Gamma$. Let $\lambda_i$ be an eigenvalue of $A$ with corresponding eigenvector $\textbf{v}_i$ for $1\leq i\leq n$ such that $\{\textbf{v}_1,\ldots,\textbf{v}_n\}$ is an orthonormal basis of $\mathbb{C}^n$. Then the spectral decomposition of $A$ is given by   
\begin{align*}
A = \lambda_1E_1 + \cdot \cdot \cdot + \lambda_nE_n, 
\end{align*}
where $E_i = \textbf{v}_i\textbf{v}_i^*~(1\leq i\leq n)$ satisfy $E_1 + \cdots + E_n = I$ and 
\[E_iE_j = \left\{ \begin{array}{cl}
     E_i & \mbox{ if }
      i = j \\
     0 & \textnormal{ otherwise. }
\end{array}\right.\]
Therefore, the spectral decomposition of the transition matrix $H(\tau )$ of $\Gamma$ is given by
\begin{align*}
H(\tau ) = \exp({-\mathbf{i}\lambda_1 \tau })E_1 + \cdot  \cdot  \cdot +\exp({-\mathbf{i}\lambda_n \tau })E_n. 
\end{align*}
 The \textit{2-adic exponential valuation} of rational numbers, denoted $\nu_2$, is a mapping $\nu_2 : \mathbb{Q} \rightarrow   \mathbb{Z}\cup\{\infty\}$ defined by
$\nu_2(0) = \infty~\mathrm{and}~
\nu_2(2^{\ell}\frac{a}{b}) = \ell,~\mathrm{where}~ a,b,\ell\in\mathbb{Z} ~\mathrm{and}~ 2\nmid ab$.
 For $\beta, \beta^\prime \in \mathbb{Q}$,
the mapping $\nu_2$ satisfies the following two properties.
\begin{enumerate}
\item[(i)] $\nu_2(\beta\beta^\prime) = \nu_2(\beta) + \nu_2(\beta^\prime)$;
\item[(ii)] $\nu_2(\beta + \beta^\prime) \geq \min(\nu_2(\beta),\nu_2(\beta^\prime))$ and the equality holds if $\nu_2(\beta) \neq \nu_2(\beta^\prime)$.
\end{enumerate}
Here, we have the convention  that $\infty + \infty = \infty + \ell = \infty$ and $\infty > \ell$ for any $\ell\in \mathbb{Z}$.

Recall that we considered the ordering 
\[1,a,a^{2},\ldots,a^{2n-1},b,ab,\ldots,a^{2n-1}b,b^{2},ab^{2},\ldots,a^{2n-1}b^{2},b^{3},ab^{3},\ldots,a^{2n-1}b^{3}\]
 for the elements of $V_{8n}$. We label these elements by $0, 1,\ldots,8n-1$ in order. With these labelling, we write the vertex set of $\text{Cay}(V_{8n},S)$ as $V_1 \cup V_2 \cup V_3 \cup V_4$, where
\begin{align*}
V_1 &=  \{0,1,\ldots,2n-1\},\hspace{1.9cm}
V_2 =  \{2n,2n+1,\ldots,4n-1\},\\
V_3 &= \{4n,4n+1,\ldots,6n-1\}~\text{ and }~
V_4 =  \{6n,6n+1,\ldots,8n-1\}.
\end{align*}
We use this partition of $V_{8n}$ in characterizing the existence of PST on normal connected Cayley graphs $\text{Cay}(V_{8n},S)$ in the next two subsections. Let $S$ be a generating subset of $V_{8n}$ such that $1 \notin S$ and $Sg = gS$ for all $g \in V_{8n}$. For odd values of $n$, let $\alpha_1, \alpha_2, \alpha_3$ and $\alpha_4$ be the eigenvalues of $\text{Cay}(V_{8n}, S)$ corresponding to the one-dimensional irreducible representations $\theta_1$, $\theta_2$, $\theta_3$ and $\theta_4$, respectively. Lemma~\ref{Lemma2.2} yields that $\alpha_i$'s are simple eigenvalues and $\alpha_1 = \lvert S \rvert$. Further, let $\beta_j$ and $\gamma_k$ be the eigenvalues corresponding to the two-dimensional irreducible  representations $\psi_j$ and $\phi_k$, respectively for $0 \leq j \leq n-1$  and $1 \leq k \leq n-1$. For even values of $n$, let $\alpha_1,\ldots,\alpha_8$ be the eigenvalues of $\text{Cay}(V_{8n}, S)$ corresponding to the one-dimensional irreducible representations $\theta_1,\ldots,\theta_8$, respectively. Lemma~\ref{Lemma2.2} yields that $\alpha_i$'s are simple eigenvalues and $\alpha_1 = \lvert S \rvert$. Further, let $\beta_j$ and $\gamma_k$ be the eigenvalues corresponding to the two-dimensional irreducible representations $\psi_j$ and $\phi_k$, respectively for $1 \leq j \leq n-1$  and $1 \leq k \leq n-1$.

\subsection{PST on normal $\text{Cay}(V_{8n}, S)$ for odd $n$}\label{subsection 3.1}
In this subsection, we present a characterization of connected normal Cayley graphs $\text{Cay}(V_{8n}, S)$ exhibiting PST. From Subsection~\ref{subsection 2.1}, we recall that the adjacency matrix of $\text{Cay}(V_{8n}, S)$ has the eigenvectors
$\textbf{u}_i,{\textbf{u}_j}^{(i)}$ and ${\textbf{v}_k}^{(i)}$ with corresponding eigenvalues $\alpha_i, \beta_j $ and $\gamma_k$, respectively,  for $1\leq i\leq 4, 0\leq j \leq n-1$ and $1\leq k\leq n-1$. Further, these eigenvectors form an orthonormal basis of $\mathbb{C}^{8n}$.

Let $J_m$ be the all-one matrix of order $m$. Then the projective matrices corresponding to the eigenvectors $\textbf{u}_1, \textbf{u}_2, \textbf{u}_3$ and $\textbf{u}_4$ are given by
\begin{align*}
E_1 &= \textbf{u}_1\textbf{u}_1^* =  \frac {1} {8n}J_{8n},
\hspace{1.85cm} E_2 = \textbf{u}_2\textbf{u}_2^* =  \frac{1}{8n} \begin{pmatrix}
J_{2n} & -J_{2n}&  J_{2n} & -J_{2n}\\
-J_{2n} & J_{2n} & -J_{2n} & J_{2n}\\
J_{2n} & -J_{2n} & J_{2n} & -J_{2n}\\
-J_{2n} &  J_{2n} & -J_{2n}& J_{2n}
\end{pmatrix},\\
E_3 &= \textbf{u}_3\textbf{u}_3^* = \frac{1}{8n}[(-1)^{u+v}]  ~\text{ and }~
E_4 = \textbf{u}_4\textbf{u}_4^* = \frac{1}{8n} [e_4(u,v)], 
\end{align*} 
respectively, where $u,v\in\{0,1,\ldots, {8n-1}\}$ and 
\[e_4(u,v)= \left\{ \begin{array}{ll} (-1)^{u+v} & \textrm{if } u,v\in V_1 \cup V_3 \textrm{ or } u,v\in V_2 \cup V_4\\
 (-1)^{u+v+1} & \textrm{if }u\in V_1 \cup V_3 , v\in V_2 \cup V_4  \textrm{ or } u\in V_2 \cup V_4 , v\in V_1 \cup V_3. \end{array}\right. \]
Note that $(-1)^{u+v}$ is the $uv$-th entry of the matrix $E_3$ and $e_{4}(u,v)$ is the $uv$-th entry of the matrix $E_4$.

The projective matrices corresponding to the eigenvectors $\textbf{u}_j^{(1)}, \textbf{u}_j^{(2)}, \textbf{u}_j^{(3)}$ and $\textbf{u}_j^{(4)}$, where $0 \leq j \leq n-1$, are given by

\begin{align*}
{E_j}^{(1)} = \textbf{u}_j^{(1)}{\textbf{u}_j^{(1)}}^* =      \frac{1}{4n}\begin{pmatrix}
                X_1 & \textbf{0} & -X_1 & \textbf{0} \\
                \textbf{0} & \textbf{0} & \textbf{0} & \textbf{0}  \\
                -X_1 & \textbf{0} & X_1 & \textbf{0}  \\
                \textbf{0} & \textbf{0} & \textbf{0} & \textbf{0}
                \end{pmatrix},\hspace{0.7cm}
{E_j}^{(2)} = \textbf{u}_j^{(2)}{\textbf{u}_j^{(2)}}^* =       \frac{1}{4n}\begin{pmatrix}
                \textbf{0} & \textbf{0} & \textbf{0} & \textbf{0} \\
                \textbf{0} & X_1 & \textbf{0} & -X_1 \\
                \textbf{0} & \textbf{0} & \textbf{0} & \textbf{0}\\
                \textbf{0} & -X_1 & \textbf{0} & X_1
                \end{pmatrix},
\end{align*}
\begin{align*}
{E_j}^{(3)} = \textbf{u}_j^{(3)}{\textbf{u}_j^{(3)}}^* = \frac{1}{4n}\begin{pmatrix}
                \textbf{0} & \textbf{0} & \textbf{0} & \textbf{0} \\
                \textbf{0} & X_2 & \textbf{0} & -X_2 \\
                \textbf{0} & \textbf{0} & \textbf{0} & \textbf{0}\\
                \textbf{0} & -X_2 & \textbf{0} & X_2
               \end{pmatrix}~\mathrm{and} ~~~
{E_j}^{(4)} = \textbf{u}_j^{(4)}{\textbf{u}_j^{(4)}}^* = \frac{1}{4n}\begin{pmatrix}
                X_2 & \textbf{0} & -X_2 & \textbf{0} \\
                \textbf{0} & \textbf{0} & \textbf{0} & \textbf{0}\\
                -X_2 & \textbf{0} & X_2 & \textbf{0} \\
                \textbf{0} & \textbf{0} & \textbf{0} & \textbf{0}
                \end{pmatrix},
\end{align*}
respectively. Here $\textbf{0}$ is the zero matrix of order $2n$;  $X_1$ and $X_2$ are circulant matrices with first row $[1, \omega^{-2j},\ldots, \omega^{-2(2n-1)j}]$ and $[1, -\omega^{2j},\ldots, -\omega^{2(2n-1)j}]$, respectively.

Further, the projective matrices corresponding to the eigenvectors $\textbf{v}_k^{(1)}, \textbf{v}_k^{(2)}, \textbf{v}_k^{(3)}$ and $\textbf{v}_k^{(4)}$, where $1 \leq k \leq n-1$, are given by
\begin{align*}                
{F_k}^{(1)} = \textbf{v}_k^{(1)}{\textbf{v}_k^{(1)}}^* = \frac{1}{4n}\begin{pmatrix}
               Y_1 & \textbf{0} & Y_1 & \textbf{0} \\
               \textbf{0} & \textbf{0} & \textbf{0} & \textbf{0}\\
               Y_1 & \textbf{0} & Y_1 & \textbf{0} \\
               \textbf{0} & \textbf{0} & \textbf{0} & \textbf{0}
               \end{pmatrix},\hspace{0.8cm}
{F_k}^{(2)} = \textbf{v}_k^{(2)}{\textbf{v}_k^{(2)}}^* = \frac{1}{4n}\begin{pmatrix}
               \textbf{0} & \textbf{0} & \textbf{0} & \textbf{0}\\
               \textbf{0} & Y_1 & \textbf{0} & Y_1 \\
               \textbf{0} & \textbf{0} & \textbf{0} & \textbf{0}\\
                \textbf{0} & Y_1 & \textbf{0} & Y_1
               \end{pmatrix},
\end{align*}
\begin{align*}
{F_k}^{(3)} = \textbf{v}_k^{(3)}{\textbf{v}_k^{(3)}}^* = \frac{1}{4n}\begin{pmatrix}
              \textbf{0} & \textbf{0} & \textbf{0} & \textbf{0}\\
              \textbf{0} & Y_2 & \textbf{0} & Y_2 \\
              \textbf{0} & \textbf{0} & \textbf{0} & \textbf{0}\\
              \textbf{0} & Y_2 & \textbf{0} & Y_2
              \end{pmatrix}~~~\mathrm{and}~~~
{F_k}^{(4)} = \textbf{v}_k^{(4)}{\textbf{v}_k^{(4)}}^* = \frac{1}{4n}\begin{pmatrix}
               Y_2 & \textbf{0} & Y_2 & \textbf{0} \\
               \textbf{0} & \textbf{0} & \textbf{0} & \textbf{0}\\
               Y_2 & \textbf{0} & Y_2 & \textbf{0} \\
               \textbf{0} & \textbf{0} & \textbf{0} & \textbf{0}
               \end{pmatrix},
\end{align*}
respectively. Here $\textbf{0}$ is the zero matrix of order $2n$; $Y_1$ and $Y_2$ are circulant matrices with first row $[1, \omega^{-k},\ldots, \omega^{-(2n-1)k}]$ and $[1, \omega^{k},\ldots, \omega^{(2n-1)k}]$, respectively.

Therefore, the transition matrix $H(\tau )$ of $\text{Cay}(V_{8n}, S)$ is given by
\begin{align}
H(\tau ) &= \exp(-\mathbf{i}\tau {\alpha}_1)E_1 + \exp(-\mathbf{i}\tau {\alpha}_2)E_2 + \exp(-\mathbf{i}\tau {\alpha}_3)E_3 + 
\exp(-\mathbf{i}\tau {\alpha}_4)E_4  \nonumber \\
&+\sum_{j=0}^{n-1}\exp(-\mathbf{i}\tau {\beta}_j)(E_j^{(1)} + E_j^{(2)} +                 E_j^{(3)} + E_j^{(4)}) 
+\sum_{k=1}^{n-1}\exp(-\mathbf{i}\tau {\gamma}_k)(F_k^{(1)} + F_k^{(2)} +                 F_k^{(3)} + F_k^{(4)}).\label{1}
\end{align}
 Using~(\ref{1}) and the eigenprojectors obtained in the preceding discussion, the $uv$-th entry of $H(\tau )$ can be easily determined.
\begin{lema}\label{lemma3.1}
Let $u$ and $v$ be two distinct vertices of $\mathrm{Cay}(V_{8n}, S)$. If $u, v \in V_1 \cup V_2$, then $\mathrm{Cay}(V_{8n}, S)$ cannot exhibit PST between $u$ and $v$.
\end{lema}
\begin{proof}
The following two cases arise.\\
Case 1. $u \in V_1, v \in V_2$ or $u \in V_2, v \in V_1$. In this case, from~(\ref{1}) we have 
\begin{align*}
H(\tau )_{uv}&= \frac{1}{8n}[\exp(-\textbf{i}\tau {\alpha}_1) - \exp(-\textbf{i}\tau {\alpha}_2) + (-1)^{u+v}\exp(-\textbf{i}\tau {\alpha}_4) + (-1)^{u+v+1} \exp(-\textbf{i}\tau {\alpha}_4)].
\end{align*}
This implies that $\mid H(\tau )_{uv} \mid \leq \frac{1}{8n} \times 4 = \frac{1}{2n} < 1$. Therefore, $\mathrm{Cay}(V_{8n}, S)$ cannot exhibit PST between the vertices $u$ and $v$. 

\noindent Case 2. $u,v \in V_1$ or $u,v \in V_2$. In this case, we have
\begin{align*}
H(\tau )_{uv} =& \frac{1}{8n} [\exp(-\textbf{i}\tau {\alpha}_1) + \exp(-\textbf{i}\tau {\alpha}_2) + (-1)^{u+v}\exp(-\textbf{i}\tau {\alpha}_3) + (-1)^{u+v}\exp(-\textbf{i}\tau {\alpha}_4)] \\
&+ \frac{1}{4n}\sum_{j=0}^{n-1}\exp(-\textbf{i}\tau {\beta}_j) [{\omega}^{2(u-v)j} + (-1)^{u+v}{\omega}^{2(v-u)j}] 
+ \frac{1}{4n}\sum_{k=1}^{n-1}\exp(-\textbf{i}\tau {\gamma}_k) [{\omega}^{(u-v)k} + {\omega}^{(v-u)k}].
\end{align*}
It can be be easily seen that $\mid{H(\tau )_{uv}} \mid  \leq 1$. Thus $\mid H(\tau )_{uv} \mid = 1$ if and only if for $0\leq j \leq n-1$ and $1 \leq k \leq n-1$, it holds that 
\begin{align*}
\exp({-\textbf{i}\tau {\alpha}_1}) &= \exp({-\textbf{i}\tau {\alpha}_2}),\\
\exp({-\textbf{i}\tau {\alpha}_1})
&= (-1)^{u+v}\exp({-\textbf{i}\tau {\alpha}_3}),\\
\exp({-\textbf{i}\tau {\alpha}_1})
&= (-1)^{u+v}\exp({-\textbf{i}\tau {\alpha}_4}),\\
\exp({-\textbf{i}\tau {\alpha}_1})
&= {\omega}^{2(u-v)j}\exp({-\textbf{i}\tau {\beta}_j}),\\
\exp({-\textbf{i}\tau {\alpha}_1})
&= (-1)^{u+v}{\omega}^{2(v-u)j}
\exp({-\textbf{i}\tau {\beta}_j}),\\
\exp({-\textbf{i}\tau {\alpha}_1})
&= {\omega}^{(u-v)k}\exp({-\textbf{i}\tau {\gamma}_k})~\mathrm{and}\\
\exp({-\textbf{i}\tau {\alpha}_1})
&= {\omega}^{(v-u)k}\exp({-\textbf{i}\tau {\gamma}_k}). 
\end{align*}

Considering the last two equations for $k = 1$, we get that $n$ divides $u-v$. Thus, the facts $u,v \in V_1$ or $u,v \in V_2$  give either $u = v$ or $u - v =\pm n$. Let $\tau  = 2 \pi T$. Then we have
\begin{align}
&({\alpha}_1-{\alpha}_2)T \in            \mathbb{Z},\label{2}\\
&({\alpha}_1-{\alpha}_3)T - \frac{u+v}{2} \in \mathbb{Z},\label{3}\\
&({\alpha}_1-{\alpha}_4)T - \frac{u+v}{2} \in \mathbb{Z},\label{4}\\
&({\alpha}_1-{\beta}_j)T \in            \mathbb{Z},\label{5}\\
&({\alpha}_1-{\beta}_j)T - \frac{u+v}{2} \in \mathbb{Z},\label{6}\\
&({\alpha}_1-{\gamma}_k)T + \frac{(u-v)k}{2n} \in\mathbb{Z}~\mathrm{and}\label{7}\\
&({\alpha}_1-{\gamma}_k) T - \frac{(u-v)k}{2n} \in \mathbb{Z}\label{8}.      
\end{align}
From conditions (\ref{2}) to (\ref{8}), we get $[3\alpha_1 + 4n\alpha_1 + 4(n-1)\alpha_1 - \alpha_2 -\alpha_3 - \alpha_4 -4\sum_{j=0}^{n-1}{\beta_j} - 4\sum_{k=1}^{n-1}{\gamma_k}]T \in \mathbb{Q}$. Since $ 0 = \mathrm{Tr}(A) = {\alpha}_1 + {\alpha} _2 + {\alpha}_3 + {\alpha}_4 + 4\sum_{j=0}^{n-1}{\beta}_j + 4\sum_{k=1}^{n-1}{\gamma}_k,$ we have that $8n{\alpha}_1T \in \mathbb{Q},$ and since ${\alpha}_1$ is a positive integer, we have $T \in \mathbb{Q}$. This implies that all the eigenvalues of the graph are rational numbers. It is well known that any rational eigenvalue of a graph is an integer. Therefore in this case the graph is integral. 

From~(\ref{5}) and (\ref{6}) it is clear that $u + v$ is even. Since $u + v$ and $u - v$ have the same parity, and $n$ is an odd integer, we cannot have $u-v =\pm n$. Thus we must have $u = v$. Therefore, $\mathrm{Cay}(V_{8n}, S)$ cannot exhibit PST between the distinct vertices $u$ and $v$. This completes the proof.
\end{proof}
Using the same technique as in Lemma~\ref{lemma3.1}, we have the following three lemmas. The proofs are omitted to avoid repetitive arguments. 

\begin{lema}\label{lemma3.2}
Let $u$ and $v$ be two distinct  vertices of $\mathrm{Cay}(V_{8n}, S)$. If $u, v \in V_1 \cup V_4$, then $\mathrm{Cay}(V_{8n}, S)$ cannot exhibit PST between $u$ and $v$.
\end{lema}

\begin{lema}\label{lemma3.3}
Let $u$ and $v$ be two distinct  vertices of $\mathrm{Cay}(V_{8n}, S)$. If $u, v \in V_2 \cup V_3$, then $\mathrm{Cay}(V_{8n}, S)$ cannot exhibit PST between $u$ and $v$.
\end{lema}

\begin{lema}\label{lemma3.4}
Let $u$ and $v$ be two distinct  vertices of $\mathrm{Cay}(V_{8n}, S)$. If $u, v \in V_3 \cup V_4$, then $\mathrm{Cay}(V_{8n}, S)$ cannot exhibit PST between $u$ and $v$.
\end{lema}
From the preceding four lemmas, observe that if $\mathrm{Cay}(V_{8n}, S)$ exhibits
PST between two distinct vertices $u$ and $v$, then either $u \in V_1, v \in V_3$ or $u \in V_3, v \in V_1$ or $u \in V_2, v \in V_4$ or $u \in V_4, v \in V_2$. In the following lemma, we give three necessary conditions for the existence of PST  on $\mathrm{Cay}(V_{8n}, S)$. We prove that these three conditions altogether are also sufficient in Lemma~\ref{lemma3.7}.
\begin{lema}\label{lemma3.5} 
Let $u$ and $v$ be two distinct  vertices of $\mathrm{Cay}(V_{8n}, S)$. If  $\mathrm{Cay}(V_{8n}, S)$ exhibits PST between $u$ and $v$, then the following three conditions hold.
\begin{enumerate}
\item[(i)] $ u - v = \pm 4n$.

\item[(ii)] All the eigenvalues of $\mathrm{Cay}(V_{8n}, S)$ are integers.

\item[(iii)] $\nu_2(\alpha_1-\beta_0) =\nu_2(\alpha_1 - \beta_j) $ for $1 \leq j \leq n-1$, and 
$\nu_2(\alpha_1 - \alpha_2)$,
$\nu_2(\alpha_1 - \alpha_3)$,
$\nu_2(\alpha_1 - \alpha_4)$, 
$\nu_2(\alpha_1 - \gamma_k)$
are strictly greater than $\nu_2(\alpha_1-\beta_0)$ for  $1 \leq k \leq n-1$.
\end{enumerate}
\end{lema}
\begin{proof}
The $uv$-th entry of the transition matrix of $\mathrm{Cay}(V_{8n}, S)$ is given by
\begin{align*}
H(\tau )_{uv} &= \frac{1}{8n}[\exp(-\textbf{i}\tau {\alpha}_1) + \exp(-\textbf{i}\tau {\alpha}_2) + (-1)^{u+v}\exp(-\textbf{i}\tau {\alpha}_3) + (-1)^{u+v}\exp(-\textbf{i}\tau {\alpha}_4)] \\ 
&+ \frac{1}{4n}\sum_{j=0}^{n-1}\exp(-\textbf{i}\tau {\beta}_j)[-{\omega}^{2(u-v)j} + (-1)^{u+v+1}{\omega}^{2(v-u)j}]
+ \frac{1}{4n}\sum_{k=1}^{n-1}\exp(-\textbf{i}\tau {\gamma}_k)[{\omega}^{(u-v)k} + {\omega}^{(v-u)k}].
\end{align*}
This implies that $\mid{H(\tau )_{uv}} \mid \leq 1$. Thus,
$ \mid H(\tau )_{uv} \mid = 1 $ if and only if for $0$ $\leq$ $j$ $\leq$ $n-1$  and $1$ $\leq$ $k$ $\leq$ $n-1$, it holds that
\begin{align*}
\exp({-\textbf{i}\tau {\alpha}_1}) &= \exp({-\textbf{i}\tau {\alpha}_2}),\\ 
\exp({-\textbf{i}\tau {\alpha}_1})
&= (-1)^{u+v} \exp({-\textbf{i}\tau {\alpha}_3}),\\
\exp({-\textbf{i}\tau {\alpha}_1})
&= (-1)^{u+v} \exp({-\textbf{i}\tau {\alpha}_4}),\\
\exp({-\textbf{i}\tau {\alpha}_1})
&= -{\omega}^{2(u-v)j}\exp({-\textbf{i}\tau {\beta}_j}),\\
\exp({-\textbf{i}\tau {\alpha}_1})
&= (-1)^{u+v+1}{\omega}^{2(v-u)j} \exp({-\textbf{i}\tau {\beta}_j}),\\
\exp({-\textbf{i}\tau {\alpha}_1})
&= {\omega}^{(u-v)k} \exp({-\textbf{i}\tau {\gamma}_k})~\mathrm{and}\\
\exp({-\textbf{i}\tau {\alpha}_1})
&= {\omega}^{(v-u)k} \exp({-\textbf{i}\tau {\gamma}_k}).
\end{align*}
From the last two equations we get that $n$ divides $u-v$. From Lemmas~\ref{lemma3.1}, \ref{lemma3.2}, \ref{lemma3.3} and \ref{lemma3.4}, note that either $u \in V_1, v \in V_3$ or $u \in V_3, v \in V_1$ or $u \in V_2, v \in V_4$ or $u \in V_4, v \in V_2$. Therefore, we have either $u - v = \pm 3n$ or  
$u - v = \pm 4n$ or $u - v = \pm 5n$. Let $\tau  = 2 \pi T$. Then we have 
\begin{align}
&({\alpha}_1-{\alpha}_2) T \in \mathbb{Z},\label{9}\\
&({\alpha}_1-{\alpha}_3) T - \frac{u+v}{2} \in \mathbb{Z},\label{10}\\
&({\alpha}_1-{\alpha}_4) T - \frac{u+v}{2} \in \mathbb{Z},\label{11}\\
&({\alpha}_1-{\beta}_j) T - \frac{1}{2}\in\mathbb{Z},\label{12}\\
&({\alpha}_1-{\beta}_j) T - \frac{u+v+1}{2}\in \mathbb{Z},\label{13}\\
&({\alpha}_1-{\gamma}_k) T + \frac{(u-v)k}{2n} \in \mathbb{Z}~\mathrm{and}\label{14}\\
&({\alpha}_1-{\gamma}_k) T - \frac{(u-v)k}{2n} \in \mathbb{Z}\label{15}.
\end{align}
From (\ref{12}) and (\ref{13}), it is clear that $u + v$ is even. So we must have $u - v =\pm 4n$. Using similar arguments as in Lemma~\ref{lemma3.1}, we find that the graph is integral. Now the conditions (\ref{9}) to (\ref{15}) can be written as
\begin{align}
&({\alpha}_1-{\alpha}_2) T \in \mathbb{Z},\label{16}\\
&({\alpha}_1-{\alpha}_3) T \in \mathbb{Z},\label{17}\\
&({\alpha}_1-{\alpha}_4) T \in \mathbb{Z},\label{18}\\
&({\alpha}_1-{\beta}_j) T - \frac{1}{2}\in\mathbb{Z}~\mathrm{and}\label{19}\\
&({\alpha}_1-{\gamma}_k) T \in \mathbb{Z}\label{20}.
\end{align}
From (\ref{19}), we have $\nu_2((\alpha_1 - \beta_j)T) = -1$ for all $j \in \{0,1, \ldots, n-1\}$. Therefore $\nu_2(\alpha_1 - \beta_j)= -1 - \nu_2(T)$ for all $j \in \{0,1, \ldots, n-1\}$. Let $\mu = -1 - \nu_2(T)$. Thus $\nu_2(\alpha_1 - \beta_j) = \mu$ for all $j \in \{0,1, \ldots, n-1\}$. From (\ref{16}), it follows that $\nu_2((\alpha_1 - \alpha_2)T) \geq 0$. Therefore, $\nu_2(\alpha_1 - \alpha_2) \geq \mu + 1$. Similarly, we find that 
$\nu_2(\alpha_1 - \alpha_3)$, 
$\nu_2(\alpha_1 - \alpha_4)$ and 
$\nu_2(\alpha_1 - \gamma_k)$ are also strictly greater than $\mu$ for all $k \in \{1, \ldots, n-1\}$. This completes the proof. 
\end{proof}

The following lemma tells us the minimum time at which $\mathrm{Cay}(V_{8n}, S)$ exhibits PST between the vertices $u$ and $v$. We use some of the techniques from~\cite{Tan} to prove the following lemma. 
\begin{lema}\label{lemma 3.6}  
If the graph $\mathrm{Cay}(V_{8n}, S)$ exhibits PST between two distinct vertices $u$ and $v$, then the minimum time at which it has PST is $\frac{\pi}{M}$, where $M = \mathrm{gcd}(\alpha_1 - \alpha \colon \alpha \in \mathrm{Spec}(\mathrm{Cay}(V_{8n}, S))\setminus \{\alpha_1\})$.
\end{lema}
\begin{proof} If the graph $\mathrm{Cay}(V_{8n}, S)$ exhibits PST between two distinct vertices $u$ and $v$, then the conditions $(i), (ii)$ and $(iii)$ of Lemma~\ref{lemma3.5} hold.
Let $M_1 = \mathrm{gcd}(\alpha_1 - \alpha_2,
 \alpha_1 - \alpha_3, \alpha_1 - \alpha_4, 
 \alpha_1 - \gamma_k \colon 1 \leq k \leq n - 1)$ and $M_2 = \mathrm{gcd}(\alpha_1 - \beta_j \colon 0 \leq j \leq n - 1)$. From the condition $(iii)$ of Lemma~\ref{lemma3.5}, we have $\nu_2(\alpha_1 - \beta_j) = \mu$. Let 
$\alpha_1 - \beta_j = {m_j}M_2$, where $\mathrm{gcd}(m_j \colon 0 \leq j \leq n - 1) = 1$. Then $\nu_2(m_j) = \mu - \nu_2(M_2) = r$, say. Suppose that $m_j = 2^{r}k_j$, where $k_j$ is odd. Then we must have $r = 0$. Thus we have $\nu_2(M_2) = \mu$. Conditions (\ref{16}), (\ref{17}), (\ref{18}) and (\ref{20}) imply that $T \in \frac{1}{M_1}\mathbb{Z}$ and condition (\ref{19}) implies that  $T \in \frac{1}{M_2}(\frac{1}{2} + \mathbb{Z})$. 

Let  
$\widetilde{T} = \{\tau>0 ~\colon~ \mathrm{Cay}(V_{8n}, S)~ \mathrm{exhibits ~PST~ between~} u~ \mathrm{and}~ v~ \mathrm{at}~ \mathrm{time}~ \tau\}$. Then we get
\begin{align}
\widetilde{T} &= \left(\frac{2\pi}{M_1}\mathbb{Z}\right)\cap \left(\frac{2\pi}{M_2}\left(\frac{1}{2} + \mathbb{Z}\right)\right) \cap \mathbb{R}^{+}
= \frac{\pi}{M_1M_2}(2M_2\mathbb{Z}\cap M_1\left(1 + 2\mathbb{Z})\right)\cap \mathbb{R}^{+}\label{21},
\end{align}
where $\mathbb{R}^{+}$ is the set of all positive real numbers.

Let $M = \mathrm{gcd}(M_1, M_2)$. Indeed, $M = \mathrm{gcd}(\alpha_1 - \alpha \colon \alpha \in \mathrm{Spec}(\mathrm{Cay}(V_{8n}, S))\setminus \{\alpha_1\})$. Write $M_1 = s_1M$ and $M_2 = s_2M$, where $\mathrm{gcd}(s_1, s_2) = 1$. From 
$\nu_2(M_1) \geq \mu + 1$ and $\nu_2(M_2) = \mu$, we get that 
$\nu_2(M) = \mu$, $s_1$ is even and $s_2$ is odd.

We now show that 
\begin{align*}
 2M_2\mathbb{Z}\cap M_1(1 + 2\mathbb{Z})  = \frac{M_1M_2}{M}(1 + 2\mathbb{Z}).
\end{align*}
Since $s_1$ is even and $s_2$ is odd, for $\ell \in \mathbb{Z}$ we find 
\[\frac{M_1M_2}{M}(1 + 2\ell )= 2M_2\left[\frac{s_1}{2}(1 + 2\ell)\right]\in 2M_2\mathbb{Z}\]
and 
\[\frac{M_1M_2}{M}(1 + 2\ell) = M_1\left[1 + 2\left(\frac{s_2-1}{2}+s_2\ell\right)\right] \in M_1(1 +2\mathbb{Z}).\]
Now assume that $z \in 2M_2\mathbb{Z}\cap M_1(1 + 2\mathbb{Z})$. Then $z = 2M_2x$
and $z = M_1(1 +2y)$ for some $x,y \in \mathbb{Z}$. Then $x$ and $y$ satisfy the diophantine equation $s_2x - s_1y = \frac{s_1}{2}$. Since $\mathrm{gcd}(s_1, s_2) = 1$, we find that 
\begin{align*}
x &= \frac{s_1}{2} + s_1\ell ~~\mathrm{and}~~y = \frac{s_2 - 1}{2} + m_2\ell,~\mathrm{where}~\ell \in \mathbb{Z}.
\end{align*}
Therefore
 \[z = 2M_2x = 2M_2\frac{m_1}{2}(1 + 2\ell) = \frac{M_1M_2}{M}(1 + 2\ell) \in \frac{M_1M_2}{M}(1 + 2\mathbb{Z}).\]
Thus
\[2M_2\mathbb{Z}\cap M_1(1 + 2\mathbb{Z})  = \frac{M_1M_2}{M}(1 + 2\mathbb{Z}).\]
Now from (\ref{21}), we get
\begin{align*}
\widetilde{T} &= \left(\frac{\pi}{M} + \frac{2\pi}{M}\mathbb{Z}\right) \cap \mathbb{R}^{+}\\
&= \left \{\frac{\pi}{M} + \frac{2\pi}{M}\ell \colon \ell \in \mathbb{N} \cup \{0\} \right \}.
\end{align*}
Therefore, the minimum time at which $\mathrm{Cay}(V_{8n}, S)$ exhibits PST is $\frac{\pi}{M}$. This completes the proof.
\end{proof}
Now we are going to prove that the conditions $(i)$, $(ii)$ and $(iii)$ of Lemma~\ref{lemma3.5} are also sufficient for the existence of PST in $\mathrm{Cay}(V_{8n}, S)$. We obtain the following lemma.
\begin{lema}\label{lemma3.7} 
Let $u, v \in V_{8n}$ and $\mathrm{Cay}(V_{8n}, S)$ satisfy the following three conditions.
\begin{enumerate}
\item[(i)] $u - v =\pm 4n$.
\item[(ii)] All the eigenvalues of $\mathrm{Cay}(V_{8n}, S)$ are integers.
\item[(iii)] $\nu_2(\alpha_1 - \beta_0) = \nu_2(\alpha_1-\beta_j)$ for $1 \leq j \leq n-1$, and $\nu_2(\alpha_1 - \alpha_2)$,
$\nu_2(\alpha_1 - \alpha_3)$,
$\nu_2(\alpha_1 - \alpha_4)$, 
$\nu_2(\alpha_1 - \gamma_k)$ are strictly greater than $\nu_2(\alpha_1 - \beta_0)$ for $1 \leq k \leq n-1$. 
\end{enumerate}
Then $\mathrm{Cay}(V_{8n}, S)$ exhibits PST between $u$ and $v$.
\end{lema}
\begin{proof}
Let $M = \mathrm{gcd}(M_1, M_2)$, where $M_1 = \mathrm{gcd}(\alpha_1 - \alpha_2,
 \alpha_1 - \alpha_3, \alpha_1 - \alpha_4, 
 \alpha_1 - \gamma_k \colon 1 \leq k \leq n - 1)$ and $M_2 = \mathrm{gcd}(\alpha_1 - \beta_j \colon 0 \leq j \leq n - 1)$.  Let  $T = \frac{1}{2M}(1 + 2 \ell)$, where $\ell \in \mathbb{N} \cup \{0\}$.  We prove that $\mathrm{Cay}(V_{8n}, S)$ exhibits PST between $u$ and $v$ at time $\tau $, where $\tau  = 2 \pi T$. 
 
Let 
$\alpha_1 - \beta_j = {m_j}M_2$, $M_1 = s_1M$ and $M_2 = s_2M$, where $\mathrm{gcd}(m_j \colon 0 \leq j \leq n - 1) = 1$ and $\mathrm{gcd}(s_1, s_2) = 1$. As in the proof of Lemma~\ref{lemma 3.6}, we find that each $m_j$ is odd, $s_1$ is even and $s_2$ is odd.

Let $\alpha_1 - \alpha_2 = k_1 M_1$ for some $k_1 \in \mathbb{Z}$. As $s_1$ is even, we have 
$$(\alpha_1 - \alpha_2)T = k_1s_1M.\frac{1}{2M}(1+2\ell) = \frac{k_1 s_1}{2}(1 + 2 \ell) \in \mathbb{Z}.$$
 Similarly, we find $({\alpha}_1-{\alpha}_3) T \in \mathbb{Z}$, $({\alpha}_1-{\alpha}_4) T \in \mathbb{Z}$ and $({\alpha}_1-{\gamma}_k) T \in \mathbb{Z}$ for $1 \leq k \leq n-1$.
In a similar way, as both $m_j$ and $s_2$ are odd, we find $({\alpha}_1-{\beta}_j) T - \frac{1}{2} \in \mathbb{Z}$ for $0 \leq j \leq n-1$. Thus for $0 \leq j \leq n-1$  and $1 \leq k \leq n-1$, we have 
\begin{align}
&({\alpha}_1-{\alpha}_2) T \in \mathbb{Z},\label{22}\\
&({\alpha}_1-{\alpha}_3) T \in \mathbb{Z},\label{23}\\
&({\alpha}_1-{\alpha}_4) T \in \mathbb{Z},\label{24}\\
&({\alpha}_1-{\beta}_j) T - \frac{1}{2}\in\mathbb{Z}~\mathrm{and}\label{25}\\
&({\alpha}_1-{\gamma}_k) T \in \mathbb{Z}\label{26}.
\end{align}
As $u -v =\pm 4n$, we see that the conditions (\ref{22}) to (\ref{26}) are equivalent to the conditions (\ref{9}) to (\ref{15}) of Lemma~\ref{lemma3.5}. Now it follows from the conditions (\ref{9}) to (\ref{15}) that $ \mid H(\tau )_{uv} \mid = 1 $. Hence 
$\mathrm{Cay}(V_{8n}, S)$ exhibits PST between $u$ and $v$ at time $\tau $.
\end{proof}


We now combine the preceding lemmas to present the main result of this subsection. 
\begin{theorem}
Let $u$ and $v$ be two distinct vertices of a connected normal Cayley graph  $\mathrm{Cay}(V_{8n}, S)$. Then $\mathrm{Cay}(V_{8n}, S)$ exhibits PST between $u$ and $v$ if and only if the following three conditions hold.
\begin{enumerate}
\item[(i)] $ u - v = \pm 4n$.

\item[(ii)] All the eigenvalues of $\mathrm{Cay}(V_{8n}, S)$ are integers.

\item[(iii)] $\nu_2(\alpha_1 - \beta_0) = \nu_2(\alpha_1-\beta_j)$ for $1 \leq j \leq n-1$, and $\nu_2(\alpha_1 - \alpha_2)$,
$\nu_2(\alpha_1 - \alpha_3)$,
$\nu_2(\alpha_1 - \alpha_4)$, 
$\nu_2(\alpha_1 - \gamma_k)$ are strictly greater than $\nu_2(\alpha_1 - \beta_0)$ for $1 \leq k \leq n-1$. 
\end{enumerate}
Furthermore, when the conditions $(i), (ii)$ and $(iii)$ hold, the minimum time at which $\mathrm{Cay}(V_{8n}, S)$ exhibits PST is $\frac{\pi}{M}$, where $M = \mathrm{gcd}(\alpha_1 - \alpha \colon \alpha \in \mathrm{Spec}(\mathrm{Cay}(V_{8n}, S))\setminus \{\alpha_1\})$.
\end{theorem}

\subsection{PST on normal $\mathrm{Cay}(V_{8n}, S)$ for even $n$}\label{subsection 3.2}
In this subsection, we present a characterization of connected normal Cayley graphs $\mathrm{Cay}(V_{8n}, S)$ exhibiting PST. From Subsection~\ref{subsection 2.2}, we recall that the adjacency matrix of $\mathrm{Cay}(V_{8n}, S)$ has the eigenvectors
$\textbf{u}_s,{\textbf{u}_j}^{(i)}$ and ${\textbf{v}_k}^{(i)}$ with corresponding eigenvalues $\alpha_s, \beta_j $ and $\gamma_k$, respectively, for  
$1\leq s \leq 8$, $1\leq i\leq 4, 1\leq j \leq n-1$ and $1\leq k\leq n-1$. Further, these eigenvectors form an orthonormal basis of $\mathbb{C}^{8n}$.

Then the projective matrices corresponding to the eigenvectors $\textbf{u}_1, \textbf{u}_3, \textbf{u}_5$ and $\textbf{u}_7$ are given by
\begin{align*}
E_1 &= \textbf{u}_1\textbf{u}_1^* =  \frac {1} {8n}J_{8n},
\hspace{5.3cm} E_3 = \textbf{u}_3\textbf{u}_3^* = \frac{1}{8n} [e_3(u,v)],\\
E_5 &= \textbf{u}_5\textbf{u}_5^* =  \frac{1}{8n} \begin{pmatrix}
J_{2n} & -J_{2n}&  J_{2n} & -J_{2n}\\
-J_{2n} & J_{2n} & -J_{2n} & J_{2n}\\
J_{2n} & -J_{2n} & J_{2n} & -J_{2n}\\
-J_{2n} &  J_{2n} & -J_{2n}& J_{2n}
\end{pmatrix}~\text{ and }~
E_7 = \textbf{u}_7\textbf{u}_7^* = \frac{1}{8n}[(-1)^{u+v}],\\
\end{align*}
respectively, where $u,v\in\{0,1\ldots, {8n-1}\}$ and 
\[e_3(u,v)= \left\{ \begin{array}{ll} (-1)^{u+v} & \textrm{if } u,v\in V_1 \cup V_3 \textrm{ or } u,v\in V_2 \cup V_4\\
 (-1)^{u+v+1} & \textrm{if }u\in V_1 \cup V_3 , v\in V_2 \cup V_4  \textrm{ or } u\in V_2 \cup V_4 , v\in V_1 \cup V_3. \end{array}\right. \]

Further, the projective matrices corresponding to the eigenvectors $\textbf{u}_2, \textbf{u}_4, \textbf{u}_6$ and $\textbf{u}_8$ are given by 
\begin{align*}
&E_2 = {\textbf{u}_2}{\textbf{u}_2}^* = \frac{1}{8n}\begin{pmatrix}
X & \textbf{i}X & -X & -\textbf{i}X \\
-\textbf{i}X & X & \textbf{i}X & -X \\
-X & -\textbf{i}X & X & \textbf{i}X \\
\textbf{i}X & -X & -\textbf{i}X & X
 \end{pmatrix},\hspace{0.6cm}
E_4 = {\textbf{u}_4}{\textbf{u}_4}^* = \frac{1}{8n}\begin{pmatrix}
Y & -\textbf{i}Y & -Y & \textbf{i}Y \\
\textbf{i}Y & Y & -\textbf{i}Y & -Y \\
-Y & \textbf{i}Y & Y & -\textbf{i}Y \\
-\textbf{i}Y & -Y & \textbf{i}Y & Y
 \end{pmatrix},\\
&E_6 = {\textbf{u}_6}{\textbf{u}_6}^* = \frac{1}{8n}\begin{pmatrix}
X & -\textbf{i}X & -X & \textbf{i}X \\
\textbf{i}X & X & -\textbf{i}X & -X \\
-X & \textbf{i}X & X & -\textbf{i}X \\
-\textbf{i}X & -X & \textbf{i}X & X
 \end{pmatrix}~\mathrm{and}~
E_8 = {\textbf{u}_8}{\textbf{u}_8}^* = \frac{1}{8n}\begin{pmatrix}
Y & \textbf{i}Y & -Y & -\textbf{i}Y \\
-\textbf{i}Y & Y & \textbf{i}Y & -Y \\
-Y & -\textbf{i}Y & Y & \textbf{i}Y \\
\textbf{i}Y & -Y & -\textbf{i}Y & Y
 \end{pmatrix},
\end{align*}
respectively, where $X = [{\textbf{i}}^{u - v}], Y = [{(-\textbf{i})}^{u - v}]$ and  $u,v\in\{0,1,\ldots, {2n-1}\}$.

The projective matrices corresponding to the eigenvectors $\textbf{u}_j^{(1)}, \textbf{u}_j^{(2)}, \textbf{u}_j^{(3)}$ and $\textbf{u}_j^{(4)}$, where $1 \leq j \leq n-1$, are given by
\begin{align*}
{E_j}^{(1)} = \textbf{u}_j^{(1)}{\textbf{u}_j^{(1)}}^* =      \frac{1}{4n}\begin{pmatrix}
                X_1 & \textbf{0} & X_1 & \textbf{0} \\
                \textbf{0} & \textbf{0} & \textbf{0} & \textbf{0}  \\
                X_1 & \textbf{0} & X_1 & \textbf{0}  \\
                \textbf{0} & \textbf{0} & \textbf{0} & \textbf{0}
                \end{pmatrix},\hspace{0.7cm}
{E_j}^{(2)} = \textbf{u}_j^{(2)}{\textbf{u}_j^{(2)}}^* =       \frac{1}{4n}\begin{pmatrix}
                \textbf{0} & \textbf{0} & \textbf{0} & \textbf{0} \\
                \textbf{0} & X_1 & \textbf{0} & X_1 \\
                \textbf{0} & \textbf{0} & \textbf{0} & \textbf{0} \\
                \textbf{0} & X_1 & \textbf{0} & X_1
                \end{pmatrix},
\end{align*}
\begin{align*}
{E_j}^{(3)} = \textbf{u}_j^{(3)}{\textbf{u}_j^{(3)}}^* = \frac{1}{4n}\begin{pmatrix}
                \textbf{0} & \textbf{0} & \textbf{0} & \textbf{0} \\
                \textbf{0} & X_2 & \textbf{0} & X_2 \\
                \textbf{0} & \textbf{0} & \textbf{0} & \textbf{0} \\
                \textbf{0} & X_2 & \textbf{0} & X_2
               \end{pmatrix}~\mathrm{and} ~~~
{E_j}^{(4)} = \textbf{u}_j^{(4)}{\textbf{u}_j^{(4)}}^* = \frac{1}{4n}\begin{pmatrix}
                X_2 & \textbf{0} & X_2 & \textbf{0} \\
                \textbf{0} & \textbf{0} & \textbf{0} & \textbf{0} \\
                X_2 & \textbf{0} & X_2 & \textbf{0} \\
                \textbf{0} & \textbf{0} & \textbf{0} & \textbf{0}
                \end{pmatrix},
\end{align*}
respectively. Here $\textbf{0}$ is the zero matrix of order $2n$; $X_1$ and $X_2$ are circulant matrices with first row $[1, \omega^{-j},\ldots, \omega^{-(2n-1)j}]$ and $[1, \omega^{j},\ldots, \omega^{(2n-1)j}]$, respectively. 

The projective matrices corresponding to the eigenvectors $\textbf{v}_k^{(1)}, \textbf{v}_k^{(2)}, \textbf{v}_k^{(3)}$ and $\textbf{v}_k^{(4)}$, where $1 \leq k \leq n-1$, are given by
\begin{align*}                
{F_k}^{(1)} = \textbf{v}_k^{(1)}{\textbf{v}_k^{(1)}}^* = \frac{1}{4n}\begin{pmatrix}
               Y_1 & \textbf{0} & -Y_1 & \textbf{0} \\
               \textbf{0} & \textbf{0} & \textbf{0} & \textbf{0} \\
               -Y_1 & \textbf{0} & Y_1 & \textbf{0} \\
               \textbf{0} & \textbf{0} & \textbf{0} & \textbf{0}
               \end{pmatrix},\hspace{1cm}
{F_k}^{(2)} = \textbf{v}_k^{(2)}{\textbf{v}_k^{(2)}}^* = \frac{1}{4n}\begin{pmatrix}
               \textbf{0} & \textbf{0} & \textbf{0} & \textbf{0} \\
               \textbf{0} & Y_1 & \textbf{0} & -Y_1 \\
               \textbf{0} & \textbf{0} & \textbf{0} & \textbf{0} \\
               \textbf{0} & -Y_1 & \textbf{0} & Y_1
               \end{pmatrix},
\end{align*}
\begin{align*}
{F_k}^{(3)} = \textbf{v}_k^{(3)}{\textbf{v}_k^{(3)}}^* = \frac{1}{4n}\begin{pmatrix}
              \textbf{0} & \textbf{0} & \textbf{0} & \textbf{0} \\
              \textbf{0} & Y_2 & \textbf{0} & -Y_2 \\
              \textbf{0} & \textbf{0} & \textbf{0} & \textbf{0} \\
              \textbf{0} & -Y_2 & \textbf{0} & Y_2
              \end{pmatrix}~~~\mathrm{and}~~~
{F_k}^{(4)} = \textbf{v}_k^{(4)}{\textbf{v}_k^{(4)}}^* = \frac{1}{4n}\begin{pmatrix}
               Y_2 & \textbf{0} & -Y_2 & \textbf{0} \\
               \textbf{0} & \textbf{0} & \textbf{0} & \textbf{0} \\
               -Y_2 & \textbf{0} & Y_2 & \textbf{0} \\
               \textbf{0} & \textbf{0} & \textbf{0} & \textbf{0}
               \end{pmatrix},
\end{align*}
respectively. Here $\textbf{0}$ is the zero matrix of order $2n$; $Y_1$ and $Y_2$ are circulant matrices with first row $[1, {\textbf{i}}^{-1}\omega^{-k},\ldots, {\textbf{i}}^{-(2n-1)}\omega^{-(2n-1)k}]$ and $[1, {\textbf{i}}^{-1}\omega^{k},\ldots, {\textbf{i}}^{-(2n-1)}\omega^{(2n-1)k}]$, respectively.            

Thus, the transition matrix $H(\tau )$ of $\mathrm{Cay}(V_{8n}, S)$ is given by
\begin{align}
H(\tau ) &= \exp(-\mathbf{i}\tau {\alpha}_1)E_1 + \exp(-\mathbf{i}\tau {\alpha}_2)E_2 + \exp(-\mathbf{i}\tau {\alpha}_3)E_3 + 
\exp(-\mathbf{i}\tau {\alpha}_4)E_4 \nonumber\\
& +\exp(-\mathbf{i}\tau {\alpha}_5)E_5 + \exp(-\mathbf{i}\tau {\alpha}_6)E_6 + \exp(-\mathbf{i}\tau {\alpha}_7)E_7 + 
\exp(-\mathbf{i}\tau {\alpha}_8)E_8 \nonumber\\
&+\sum_{j=1}^{n-1}\exp(-\mathbf{i}\tau {\beta}_j)(E_j^{(1)} + E_j^{(2)} +                 E_j^{(3)} + E_j^{(4)}) 
+\sum_{k=1}^{n-1}\exp(-\mathbf{i}\tau {\gamma}_k)(F_k^{(1)} + F_k^{(2)} +                 F_k^{(3)} + F_k^{(4)})\label{27}.
\end{align}
 Using~(\ref{27}) and the eigenprojectors obtained in the preceding discussion, the $uv$-th entry of $H(\tau )$ can be easily determined.
\begin{lema}\label{lemma3.8}
If $u \in V_1, v \in V_2$ $\mathrm{or}$ $u \in V_2, v \in V_1$, then $\mathrm{Cay}(V_{8n}, S)$ cannot exhibit PST between the vertices $u$ and $v$.
\end{lema}
\begin{proof}
It is enough to consider the case that $u \in V_1, v \in V_2$. In this case, from ~(\ref{27}) we have
\begin{align*}
H(\tau )_{uv} = &\frac{1}{8n} (\exp(-\textbf{i}\tau {\alpha}_1) + {\textbf{i}}^{u - v + 1}\exp(-\textbf{i}\tau {\alpha}_2) + (-1)^{u+v+1}\exp(-\textbf{i}\tau {\alpha}_3) +(-\textbf{i})^{u-v+1}\exp(-\textbf{i}\tau {\alpha}_4)\\ 
&-\exp(-\textbf{i}\tau {\alpha}_5) + {\textbf{i}}^{u-v+3}\exp(-\textbf{i}\tau {\alpha}_6) + (-1)^{u+v}\exp(-\textbf{i}\tau {\alpha}_7) +(-\textbf{i})^{u-v+3}\exp(-\textbf{i}\tau {\alpha}_8)).
\end{align*}
Therefore $\mid H(\tau )_{uv} \mid \leq \frac{1}{8n} \times 8 = \frac{1}{n} < 1$. Hence $\mathrm{Cay}(V_{8n}, S)$  cannot exhibit PST between $u$ and $v$.
\end{proof}
Using the same technique as in Lemma~\ref{lemma3.8}, we have the following three lemmas. The proofs are omitted to avoid repetitive arguments. 

\begin{lema}\label{lemma3.9}
If $u \in V_1, v \in V_4$ $\mathrm{or}$ $u \in V_4, v \in V_1$, then $\mathrm{Cay}(V_{8n}, S)$ cannot exhibit PST between the vertices $u$ and $v$.
\end{lema}

\begin{lema}\label{lemma3.10}
If $u \in V_2, v \in V_3$ $\mathrm{or}$ $u \in V_3, v \in V_2$, then $\mathrm{Cay}(V_{8n}, S)$ cannot exhibit PST between the vertices $u$ and $v$.
\end{lema}

\begin{lema}\label{lemma3.11}
If $u \in V_3, v \in V_4$ $\mathrm{or}$ $u \in V_4, v \in V_3$, then $\mathrm{Cay}(V_{8n}, S)$ cannot exhibit PST between the vertices $u$ and $v$.
\end{lema}
From the preceding three lemmas, we observe that if $\mathrm{Cay}(V_{8n}, S)$ exhibits PST between the distinct vertices $u$ and $v$, then it is necessary that $u,v \in V_1$ or $u,v \in V_2$ or $u,v \in V_3$ or $u,v \in V_4$ or $u \in V_1, v \in V_3$ or $u \in V_3, v \in V_1$ or $u \in V_2, v \in V_4$ or $u \in V_4, v \in V_2$. If $u$ and $v$ satisfy one of these conditions, then we are going to determine a few more necessary and sufficient conditions for the existence of PST   
in $\mathrm{Cay}(V_{8n}, S)$ between $u$ and $v$.

Our next discussion is distinguished into two cases according as $n \equiv 0 ~(\mathrm{mod}~ 4)$ or $n \equiv 2 ~(\mathrm{mod}~ 4)$. Let us first consider the case that $n \equiv 0 ~(\mathrm{mod}~ 4)$. We have the following lemma.
\begin{lema}\label{lemma3.12}
  Let $u,v \in V_1$ or $u,v \in V_2$ or $u,v \in V_3$ or $u,v \in V_4$ and $n \equiv 0 ~(\mathrm{mod}~ 4)$. Then $\mathrm{Cay}(V_{8n}, S)$ exhibits PST between the distinct vertices $u$ and $v$  if and only if  the following three conditions hold.
\begin{enumerate}
\item[(i)] $u - v =\pm n$.
\item[(ii)] All the eigenvalues of $\mathrm{Cay}(V_{8n}, S)$ are integers.
\item[(iii)] $\nu_2(\alpha_1 - \beta_{2j^\prime - 1}) = \nu_2(\alpha_1 - \gamma_{2k^\prime - 1})$, and $\nu_2(\alpha_1 - \alpha_2)$,
$\nu_2(\alpha_1 - \alpha_3)$,
$\nu_2(\alpha_1 - \alpha_4)$,
$\nu_2(\alpha_1 - \alpha_5)$,
$\nu_2(\alpha_1 - \alpha_6)$,
$\nu_2(\alpha_1 - \alpha_7)$,
$\nu_2(\alpha_1 - \alpha_8)$,
$\nu_2(\alpha_1 - \beta_{2j^\prime})$,
$\nu_2(\alpha_1 - \gamma_{2k^\prime})$ 
are strictly greater than $\nu_2(\alpha_1 - \beta_1)$ for 
$1 \leq j^\prime \leq \frac{n}{2}$  and $1 \leq k^\prime \leq \frac{n}{2}$. 
\end{enumerate}
Furthermore, when the conditions $(i), (ii)$ and $(iii)$ hold, the minimum time at which $\mathrm{Cay}(V_{8n}, S)$ exhibits PST is $\frac{\pi}{M}$, where $M = \mathrm{gcd}(\alpha_1 - \alpha \colon \alpha \in \mathrm{Spec}(\mathrm{Cay}(V_{8n}, S))\setminus \{\alpha_1\})$.
\end{lema}
\begin{proof}
If $u,v \in V_1$ or $u,v \in V_2$ or $u,v \in V_3$ or $u,v \in V_4$, then 
\begin{align*}
H(\tau )_{uv} = &\frac{1}{8n} [\exp(-\textbf{i}\tau {\alpha}_1) + {\textbf{i}}^{u-v}\exp(-\textbf{i}\tau {\alpha}_2) + (-1)^{u+v}\exp(-\textbf{i}\tau {\alpha}_3) + (-\textbf{i})^{u-v}\exp(-\textbf{i}\tau {\alpha}_4)\\
&+\exp(-\textbf{i}\tau {\alpha}_5) + {\textbf{i}}^{u-v}\exp(-\textbf{i}\tau {\alpha}_6) + (-1)^{u+v}\exp(-\textbf{i}\tau {\alpha}_7) + (-\textbf{i})^{u-v}\exp(-\textbf{i}\tau {\alpha}_8)]\\
&+ \frac{1}{4n}\sum_{j=1}^{n-1}\exp(-\textbf{i}\tau {\beta}_j) [{\omega}^{(u-v)j} + {\omega}^{-(u-v)j}] 
+ \frac{1}{4n}\sum_{k=1}^{n-1}\exp(-\textbf{i}\tau {\gamma}_k) [{\textbf{i}}^{u-v}{\omega}^{(u-v)k} + {\textbf{i}}^{u-v}{\omega}^{-(u-v)k}].
\end{align*}
Thus we have $\mid{H(\tau )_{uv}} \mid  \leq 1$. Therefore $\mid H(\tau )_{uv} \mid = 1$ if and only if for $1\leq j \leq n-1$ and $1 \leq k \leq n-1$, it holds that
\begin{align*}
\exp({-\textbf{i}\tau {\alpha}_1}) &= {\textbf{i}}^{u-v}\exp({-\textbf{i}\tau {\alpha}_2}),\\
\exp({-\textbf{i}\tau {\alpha}_1})
&= (-1)^{u+v}\exp({-\textbf{i}\tau {\alpha}_3}),\\
\exp({-\textbf{i}\tau {\alpha}_1})
&= (-\textbf{i})^{u-v}\exp({-\textbf{i}\tau {\alpha}_4}),\\
\exp({-\textbf{i}\tau {\alpha}_1}) &= \exp({-\textbf{i}\tau {\alpha}_5}),\\
\exp({-\textbf{i}\tau {\alpha}_1})
&= \textbf{i}^{u-v}\exp({-\textbf{i}\tau {\alpha}_6}),\\
\exp({-\textbf{i}\tau {\alpha}_1})
&= (-1)^{u+v}\exp({-\textbf{i}\tau {\alpha}_7}),\\
\exp({-\textbf{i}\tau {\alpha}_1})
&= (-\textbf{i})^{u-v}\exp({-\textbf{i}\tau {\alpha}_8}),\\
\exp({-\textbf{i}\tau {\alpha}_1})
&= {\omega}^{(u-v)j}\exp({-\textbf{i}\tau {\beta}_j}),\\
\exp({-\textbf{i}\tau {\alpha}_1})
&= {\omega}^{-(u-v)j}
\exp({-\textbf{i}\tau {\beta}_j}),\\
\exp({-\textbf{i}\tau {\alpha}_1})
&= {\textbf{i}}^{u-v}{\omega}^{(u-v)k}\exp({-\textbf{i}\tau {\gamma}_k})~\mathrm{and}~\\
\exp({-\textbf{i}\tau {\alpha}_1})
&= {\textbf{i}}^{u-v}{\omega}^{-(u-v)k}\exp({-\textbf{i}\tau {\gamma}_k}). 
\end{align*}
From the last two equations, we get $u - v = \pm n$. Let $\tau  = 2 \pi T$. Then the preceding equations imply
\begin{align}
&({\alpha}_1-{\alpha}_2)T \in \mathbb{Z},\label{28}\\
&({\alpha}_1-{\alpha}_3)T \in \mathbb{Z},\label{29}\\
&({\alpha}_1-{\alpha}_4)T \in \mathbb{Z},\label{30}\\
&({\alpha}_1-{\alpha}_5)T \in \mathbb{Z},\label{31}\\
&({\alpha}_1-{\alpha}_6)T \in \mathbb{Z},\label{32}\\
&({\alpha}_1-{\alpha}_7)T \in \mathbb{Z},\label{33}\\
&({\alpha}_1-{\alpha}_8)T \in \mathbb{Z},\label{34}\\
&({\alpha}_1-{\beta}_j)T - \frac{j}{2} \in \mathbb{Z}~\mathrm{and}~\label{35}\\
&({\alpha}_1-{\gamma}_k)T - \frac{k}{2} \in\mathbb{Z}\label{36}.     
\end{align}
Using similar arguments as in Lemma~\ref{lemma3.1}, it can be proved that the graph is integral.  From the conditions (\ref{28}) to (\ref{36}), proceeding as in Lemma~\ref{lemma3.5}, we have
$\nu_2(\alpha_1 - \beta_{2j^\prime - 1})=\nu_2(\alpha_1 - \gamma_{2k^\prime - 1})$, and $\nu_2(\alpha_1 - \alpha_2)$,
$\nu_2(\alpha_1 - \alpha_3)$,
$\nu_2(\alpha_1 - \alpha_4)$,
$\nu_2(\alpha_1 - \alpha_5)$,
$\nu_2(\alpha_1 - \alpha_6)$,
$\nu_2(\alpha_1 - \alpha_7)$,
$\nu_2(\alpha_1 - \alpha_8)$,
$\nu_2(\alpha_1 - \beta_{2j^\prime})$,
$\nu_2(\alpha_1 - \gamma_{2k^\prime})$ are strictly greater than $\nu_2(\alpha_1 - \beta_1)$ for 
$1 \leq j^\prime \leq \frac{n}{2}$  and $1 \leq k^\prime \leq \frac{n}{2}$. 

The converse part of the proof is similar to that of Lemma~\ref{lemma3.7}. Further, the process of determination of the minimum time at which $\mathrm{Cay}(V_{8n}, S)$ exhibits PST is also similar to the proof of Lemma~\ref{lemma 3.6}. Therefore the details are omitted. 
\end{proof}

The next lemma considers the case  $n \equiv 2 ~(\mathrm{mod}~ 4)$ and $u,v \in V_1$ or $u,v \in V_2$ or $u,v \in V_3$ or $u,v \in V_4$.  The proof is similar to that of Lemma~\ref{lemma3.12}.

\begin{lema}\label{lemma3.13}
Let $u,v \in V_1$ or $u,v \in V_2$ or $u,v \in V_3$ or $u,v \in V_4$ and $n \equiv 2 ~(\mathrm{mod}~ 4)$. Then $\mathrm{Cay}(V_{8n}, S)$ exhibits PST between the distinct vertices $u$ and $v$ if and only if the following three conditions hold.
\begin{enumerate}
\item[(i)] $u - v =\pm  n$.
\item[(ii)] All the eigenvalues of $\mathrm{Cay}(V_{8n}, S)$ are integers.
\item[(iii)] $\nu_2(\alpha_1 - \alpha_2) = \nu_2(\alpha_1 - \alpha_4) = \nu_2(\alpha_1 - \alpha_6) = \nu_2(\alpha_1 - \alpha_8) = \nu_2(\alpha_1 - \beta_{2j^\prime - 1}) = \nu_2(\alpha_1 - \gamma_{2k^\prime})$, and $\nu_2(\alpha_1 - \alpha_3)$,
$\nu_2(\alpha_1 - \alpha_5)$,
$\nu_2(\alpha_1 - \alpha_7)$,
$\nu_2(\alpha_1 - \beta_{2j^\prime})$,  
$\nu_2(\alpha_1 - \gamma_{2k^\prime - 1})$ are 
strictly greater  than $\nu_2(\alpha_1 - \alpha_2)$ for 
$1 \leq j^\prime \leq \frac{n}{2}$  and $1 \leq k^\prime \leq \frac{n}{2}$. 
\end{enumerate}
Furthermore, when the conditions $(i), (ii)$ and $(iii)$ hold, the minimum time at which $\mathrm{Cay}(V_{8n}, S)$ exhibits PST is $\frac{\pi}{M}$, where $M = \mathrm{gcd}(\alpha_1 - \alpha \colon \alpha \in \mathrm{Spec}(\mathrm{Cay}(V_{8n}, S))\setminus \{\alpha_1\})$.
\end{lema}
 In the following lemma, we consider the cases   when $u \in V_1, v \in V_3$ or $ u \in V_3, v \in V_1$ or $u \in V_2, v \in V_4$ or $u \in V_4, v \in V_2$ for the existence of PST between $u$ and $v$.

\begin{lema}\label{lemma3.14}
Let $u \in V_1, v \in V_3$ or $u \in V_3, v \in V_1$ or $u \in V_2, v \in V_4$ or $u \in V_4, v \in V_2$. Then $\mathrm{Cay}(V_{8n}, S)$ exhibits PST between the  vertices $u$ and $v$ if and only if all the eigenvalues of $\mathrm{Cay}(V_{8n}, S)$ are integers and one of the following three conditions hold.
\begin{enumerate}
\item[1.(i)] $u - v =\pm 4n$.
\item[(ii)] $\nu_2(\alpha_1 - \alpha_2)=\nu_2(\alpha_1 - \alpha_4)=\nu_2(\alpha_1 - \alpha_6)=\nu_2(\alpha_1 - \alpha_8)=\nu_2(\alpha_1 - \gamma_k)$, and 
$\nu_2(\alpha_1 - \alpha_3)$,
$\nu_2(\alpha_1 - \alpha_5)$,
$\nu_2(\alpha_1 - \alpha_7)$, $\nu_2(\alpha_1 - \beta_j)$ are strictly greater than $\nu_2(\alpha_1 - \alpha_2)$ for $1 \leq  j \leq n-1$  
and $1 \leq k \leq n-1$.  
\item[2. (i)] $n \equiv 0 ~(\mathrm{mod}~ 4)$ and $u - v\in\{ \pm 3n, \pm 5n$\}.

\item[(ii)] $\nu_2(\alpha_1 - \alpha_2)=\nu_2(\alpha_1 - \alpha_4)=\nu_2(\alpha_1 - \alpha_6)=\nu_2(\alpha_1 - \alpha_8)=\nu_2(\alpha_1 - \beta_{2j^\prime - 1})=\nu_2(\alpha_1 - \gamma_{2k^\prime})$, and $\nu_2(\alpha_1 - \alpha_3)$,
$\nu_2(\alpha_1 - \alpha_5)$,
$\nu_2(\alpha_1 - \alpha_7)$,
$\nu_2(\alpha_1 - \beta_{2j^\prime})$, 
$\nu_2(\alpha_1 - \gamma_{2k^\prime - 1})$ are 
strictly greater  than $\nu_2(\alpha_1 - \alpha_2)$ for 
$1 \leq j^\prime \leq \frac{n}{2}$  and $1 \leq k^\prime \leq \frac{n}{2}$. 
\item[3. (i)] $n \equiv 2 ~(\mathrm{mod}~ 4)$ and $u - v\in\{ \pm 3n, \pm 5n$\}.

\item[(ii)] $\nu_2(\alpha_1 - \beta_{2j^\prime - 1})=\nu_2(\alpha_1 - \gamma_{2k^\prime - 1})$, and 
$\nu_2(\alpha_1 - \alpha_2)$,
$\nu_2(\alpha_1 - \alpha_3)$,
$\nu_2(\alpha_1 - \alpha_4)$,
$\nu_2(\alpha_1 - \alpha_5)$,
$\nu_2(\alpha_1 - \alpha_6)$,
$\nu_2(\alpha_1 - \alpha_7)$,
$\nu_2(\alpha_1 - \alpha_8)$,
$\nu_2(\alpha_1 - \beta_{2j^\prime})$, 
$\nu_2(\alpha_1 - \gamma_{2k^\prime})$ are 
strictly 
greater than $\nu_2(\alpha_1 - \beta_1)$ for 
$1 \leq j^\prime \leq \frac{n}{2}$  and $1 \leq k^\prime \leq \frac{n}{2}$. 
\end{enumerate}
Furthermore, if PST exists in $\mathrm{Cay}(V_{8n}, S)$, then  the minimum time at which it exhibits  PST is $\frac{\pi}{M}$, where $M = \mathrm{gcd}(\alpha_1 - \alpha \colon \alpha \in \mathrm{Spec}(\mathrm{Cay}(V_{8n}, S))\setminus \{\alpha_1\})$.
\end{lema}
\begin{proof}
If $u \in V_1, v \in V_3$ or $u \in V_3, v \in V_1$ or $u \in V_2, v \in V_4$ or $u \in V_4, v \in V_2$, then
\begin{align*}
H(\tau )_{uv} &= \frac{1}{8n} [\exp(-\textbf{i}\tau {\alpha}_1) - {\textbf{i}}^{u-v}\exp(-\textbf{i}\tau {\alpha}_2) + (-1)^{u+v}\exp(-\textbf{i}\tau {\alpha}_3) - (-\textbf{i})^{u-v}\exp(-\textbf{i}\tau {\alpha}_4)\\
&+\exp(-\textbf{i}\tau {\alpha}_5) - {\textbf{i}}^{u-v}\exp(-\textbf{i}\tau {\alpha}_6) + (-1)^{u+v}\exp(-\textbf{i}\tau {\alpha}_7) - (-\textbf{i})^{u-v}\exp(-\textbf{i}\tau {\alpha}_8)]\\
&+ \frac{1}{4n}\sum_{j=1}^{n-1}\exp(-\textbf{i}\tau {\beta}_j) [{\omega}^{(u-v)j} + {\omega}^{-(u-v)j}] 
+ \frac{1}{4n}\sum_{k=1}^{n-1}\exp(-\textbf{i}\tau {\gamma}_k) [-{\textbf{i}}^{u-v}{\omega}^{(u-v)k} - {\textbf{i}}^{u-v}{\omega}^{-(u-v)k}].
\end{align*}
Therefore $\mid{H(\tau )_{uv}} \mid  \leq 1$. Thus $\mid H(\tau )_{uv} \mid = 1$ if and only if for $1\leq j \leq n-1$ and $1 \leq k \leq n-1$, it holds that
\begin{align}
\exp({-\textbf{i}\tau {\alpha}_1}) &= -{\textbf{i}}^{u-v}\exp({-\textbf{i}\tau {\alpha}_2}),\label{37}\\
\exp({-\textbf{i}\tau {\alpha}_1})
&= (-1)^{u+v}\exp({-\textbf{i}\tau {\alpha}_3}),\label{38}\\
\exp({-\textbf{i}\tau {\alpha}_1})
&= -(-\textbf{i})^{u-v}\exp({-\textbf{i}\tau {\alpha}_4}),\label{39}\\
\exp({-\textbf{i}t{\alpha}_1}) &= \exp({-\textbf{i}\tau {\alpha}_5}),\label{40}\\
\exp({-\textbf{i}\tau {\alpha}_1})
&= -\textbf{i}^{u-v}\exp({-\textbf{i}\tau {\alpha}_6}),\label{41}\\
\exp({-\textbf{i}\tau {\alpha}_1})
&= (-1)^{u+v}\exp({-\textbf{i}\tau {\alpha}_7}),\label{42}\\
\exp({-\textbf{i}\tau {\alpha}_1})
&= -(-\textbf{i})^{u-v}\exp({-\textbf{i}\tau {\alpha}_8}),\label{43}\\
\exp({-\textbf{i}\tau {\alpha}_1})
&= {\omega}^{(u-v)j}\exp({-\textbf{i}\tau {\beta}_j}),\label{44}\\
\exp({-\textbf{i}\tau {\alpha}_1})
&= {\omega}^{-(u-v)j}
\exp({-\textbf{i}\tau {\beta}_j}),\label{45}\\
\exp({-\textbf{i}\tau {\alpha}_1})
&= -{\textbf{i}}^{u-v}{\omega}^{(u-v)k}\exp({-\textbf{i}\tau {\gamma}_k})~\mathrm{and}~\label{46}\\
\exp({-\textbf{i}\tau {\alpha}_1})
&= -{\textbf{i}}^{u-v}{\omega}^{-(u-v)k}\exp({-\textbf{i}\tau {\gamma}_k})\label{47}. 
\end{align}
For $k=1$, the last two equations give $u - v \in \{ \pm 3n,  \pm 4n ,  \pm 5n$\}. Let us first consider the case that $u - v = \pm 4n$. Suppose that $\tau  = 2 \pi T$. Then the equations (\ref{37}) to (\ref{47}) reduce to
\begin{align}
&({\alpha}_1-{\alpha}_2)T -\frac{1}{2} \in \mathbb{Z},\label{48}\\
&({\alpha}_1-{\alpha}_3)T \in \mathbb{Z},\label{49}\\
&({\alpha}_1-{\alpha}_4)T -\frac{1}{2}\in \mathbb{Z},\label{50}\\
&({\alpha}_1-{\alpha}_5)T \in \mathbb{Z},\label{51}\\
&({\alpha}_1-{\alpha}_6)T -\frac{1}{2}\in \mathbb{Z},\label{52}\\
&({\alpha}_1-{\alpha}_7)T \in \mathbb{Z},\label{53}\\
&({\alpha}_1-{\alpha}_8)T -\frac{1}{2}\in \mathbb{Z},\label{54}\\
&({\alpha}_1-{\beta}_j)T \in \mathbb{Z}~\mathrm{and}~\label{55}\\
&({\alpha}_1-{\gamma}_k)T - \frac{1}{2} \in\mathbb{Z}\label{56}.     
\end{align}
Using similar argument as in Lemma~\ref{lemma3.1}, we find that the graph is integral. Further, proceeding as in Lemma~\ref{lemma3.5}, we get
$\nu_2(\alpha_1 - \alpha_2)=\nu_2(\alpha_1 - \alpha_4)=\nu_2(\alpha_1 - \alpha_6)=\nu_2(\alpha_1 - \alpha_8)=\nu_2(\alpha_1 - \gamma_k)$, and 
$\nu_2(\alpha_1 - \alpha_3)$,
$\nu_2(\alpha_1 - \alpha_5)$,
$\nu_2(\alpha_1 - \alpha_7)$
and $\nu_2(\alpha_1 - \beta_j)$ are strictly greater than $\nu_2(\alpha_1 - \alpha_2)$ for $1 \leq  j \leq n-1$  
and $1 \leq k \leq n-1$.  

Now we consider the case that $n \equiv 0 ~(\mathrm{mod}~ 4)$ and  $u - v\in\{ \pm 3n, \pm 5n\}$. Let $\tau  = 2\pi T$. Then the equations (\ref{37}) to (\ref{47}) reduce to 
\begin{align}
&({\alpha}_1-{\alpha}_2)T -\frac{1}{2} \in \mathbb{Z},\label{57}\\
&({\alpha}_1-{\alpha}_3)T \in \mathbb{Z},\label{58}\\
&({\alpha}_1-{\alpha}_4)T -\frac{1}{2}\in \mathbb{Z},\label{59}\\
&({\alpha}_1-{\alpha}_5)T \in \mathbb{Z},\label{60}\\
&({\alpha}_1-{\alpha}_6)T -\frac{1}{2}\in \mathbb{Z},\label{61}\\
&({\alpha}_1-{\alpha}_7)T \in \mathbb{Z},\label{62}\\
&({\alpha}_1-{\alpha}_8)T -\frac{1}{2}\in \mathbb{Z},\label{63}\\
&({\alpha}_1-{\beta}_j)T -\frac{j}{2}\in \mathbb{Z}~\mathrm{and}~\label{64}\\
&({\alpha}_1-{\gamma}_k)T - \frac{1}{2} -\frac{k}{2} \in\mathbb{Z}\label{65}.     
\end{align}
 Now as in the previous case, we find from these conditions  that  the graph is integral. Further, \linebreak[4]
$\nu_2(\alpha_1 - \alpha_2)=\nu_2(\alpha_1 - \alpha_4)=\nu_2(\alpha_1 - \alpha_6)=\nu_2(\alpha_1 - \alpha_8)=\nu_2(\alpha_1 - \beta_{2j^\prime - 1})=\nu_2(\alpha_1 - \gamma_{2k^\prime})$, and $\nu_2(\alpha_1 - \alpha_3)$,
$\nu_2(\alpha_1 - \alpha_5)$,
$\nu_2(\alpha_1 - \alpha_7)$,
$\nu_2(\alpha_1 - \beta_{2j^\prime})$, 
$\nu_2(\alpha_1 - \gamma_{2k^\prime - 1})$ are 
strictly greater  than $\nu_2(\alpha_1 - \alpha_2)$ for 
$1 \leq j^\prime \leq \frac{n}{2}$  and $1 \leq k^\prime \leq \frac{n}{2}$.

Similarly, for the case that $n \equiv 2 ~(\mathrm{mod}~ 4)$ and $u - v\in\{ \pm 3n, \pm 5n\}$ also, we find that  the graph is integral. Further,  
$\nu_2(\alpha_1 - \beta_{2j^\prime - 1})=\nu_2(\alpha_1 - \gamma_{2k^\prime - 1})$, and 
$\nu_2(\alpha_1 - \alpha_2)$,
$\nu_2(\alpha_1 - \alpha_3)$,
$\nu_2(\alpha_1 - \alpha_4)$,
$\nu_2(\alpha_1 - \alpha_5)$,
$\nu_2(\alpha_1 - \alpha_6)$,
$\nu_2(\alpha_1 - \alpha_7)$,
$\nu_2(\alpha_1 - \alpha_8)$,
$\nu_2(\alpha_1 - \beta_{2j^\prime})$,
$\nu_2(\alpha_1 - \gamma_{2k^\prime})$ are 
strictly
greater than $\nu_2(\alpha_1 - \beta_1)$ for 
$1 \leq j^\prime \leq \frac{n}{2}$  and $1 \leq k^\prime \leq \frac{n}{2}$. 

The rest of the proof is similar to that of Lemma~\ref{lemma 3.6} and Lemma~\ref{lemma3.7}. Hence the details are omitted.
\end{proof}

For convenience, we say that $\mathrm{Cay}(V_{8n}, S)$ is of \textbf{Type 1} if its eigenvalues are integers and satisfy the conditions that $\nu_2(\alpha_1 - \beta_{2j^\prime - 1})=\nu_2(\alpha_1 - \gamma_{2k^\prime - 1})$, and $\nu_2(\alpha_1 - \alpha_2)$,
$\nu_2(\alpha_1 - \alpha_3)$,
$\nu_2(\alpha_1 - \alpha_4)$,
$\nu_2(\alpha_1 - \alpha_5)$,
$\nu_2(\alpha_1 - \alpha_6)$,
$\nu_2(\alpha_1 - \alpha_7)$,
$\nu_2(\alpha_1 - \alpha_8)$,
$\nu_2(\alpha_1 - \beta_{2j^\prime})$,
$\nu_2(\alpha_1 - \gamma_{2k^\prime})$ 
are strictly greater than $\nu_2(\alpha_1 - \beta_1)$ for 
$1 \leq j^\prime \leq \frac{n}{2}$  and $1 \leq k^\prime \leq \frac{n}{2}$. 

Similarly, we say that $\mathrm{Cay}(V_{8n}, S)$ is of \textbf{Type 2} if its eigenvalues are integers and satisfy the conditions that $\nu_2(\alpha_1 - \alpha_2)=\nu_2(\alpha_1 - \alpha_4)=\nu_2(\alpha_1 - \alpha_6)=\nu_2(\alpha_1 - \alpha_8)=\nu_2(\alpha_1 - \beta_{2j^\prime - 1})=\nu_2(\alpha_1 - \gamma_{2k^\prime})$, and $\nu_2(\alpha_1 - \alpha_3)$,
$\nu_2(\alpha_1 - \alpha_5)$,
$\nu_2(\alpha_1 - \alpha_7)$,
$\nu_2(\alpha_1 - \beta_{2j^\prime})$,
$\nu_2(\alpha_1 - \gamma_{2k^\prime - 1})$ are 
strictly greater  than $\nu_2(\alpha_1 - \alpha_2)$ for 
$1 \leq j^\prime \leq \frac{n}{2}$  and $1 \leq k^\prime \leq \frac{n}{2}$.

Finally, we say that $\mathrm{Cay}(V_{8n}, S)$ is of \textbf{Type 3} if its eigenvalues are integers and satisfy the conditions that $\nu_2(\alpha_1 - \alpha_2)=\nu_2(\alpha_1 - \alpha_4)=\nu_2(\alpha_1 - \alpha_6)=\nu_2(\alpha_1 - \alpha_8)=\nu_2(\alpha_1 - \gamma_k)$, and 
$\nu_2(\alpha_1 - \alpha_3)$,
$\nu_2(\alpha_1 - \alpha_5)$,
$\nu_2(\alpha_1 - \alpha_7)$,  $\nu_2(\alpha_1 - \beta_j)$ are strictly greater than $\nu_2(\alpha_1 - \alpha_2)$ for $1 \leq  j \leq n-1$  
and $1 \leq k \leq n-1$.

Now, combining Lemmas~\ref{lemma3.12},~\ref{lemma3.13} and ~\ref{lemma3.14}, we present the  main result of this subsection.

\begin{theorem}
Let $u$ and $v$ be two distinct vertices of a connected normal Cayley graph  $\mathrm{Cay}(V_{8n}, S)$. Then PST exists in $\mathrm{Cay}(V_{8n}, S)$ between $u$ and $v$ if and only if  one of  the following three conditions hold.
\begin{enumerate}
\item $\mathrm{Cay}(V_{8n}, S)$ is of Type 1, and either 
\begin{enumerate}
\item   $u,v \in V_1$ or  $u,v \in V_2$ or $u,v \in V_3$ or $u,v \in V_4$, and $n \equiv 0 ~(\mathrm{mod}~ 4), u-v =\pm n$; or
\item   $u \in V_1, v \in V_3$ or $u \in V_3, v \in V_1$ or $u \in V_2, v \in V_4$ or $u \in V_4, v \in V_2$, and $n \equiv 2 ~(\mathrm{mod}~ 4)$, $u-v \in\{\pm 3n, \pm 5n\}$.
\end{enumerate}

\item $\mathrm{Cay}(V_{8n}, S)$ is of Type 2, and either 
\begin{enumerate}
\item  $u,v \in V_1$ or  $u,v \in V_2$ or $u,v \in V_3$ or $u,v \in V_4$, and $n \equiv 2 ~(\mathrm{mod}~ 4), u-v =\pm n$; or
\item $u \in V_1, v \in V_3$ or $u \in V_3, v \in V_1$ or $u \in V_2, v \in V_4$ or $u \in V_4, v \in V_2$, and $n \equiv 0 ~(\mathrm{mod}~ 4)$, $u-v \in\{\pm 3n, \pm 5n\}$. 
\end{enumerate}
\item $\mathrm{Cay}(V_{8n}, S)$ is of Type 3 and $u-v=\pm 4n$.
\end{enumerate}
Further, if  $\mathrm{Cay}(V_{8n}, S)$ exhibits PST between $u$ and $v$, then the minimum time at which it exhibits PST is $\frac{\pi}{M}$, where $M = \mathrm{gcd}(\alpha - \alpha_1 \colon \alpha \in \mathrm{Spec}(\mathrm{Cay}(V_{8n}, S))\setminus \{\alpha_1\})$.

\end{theorem}

\subsection*{Acknowledgement} The first author acknowledges the
support provided by the Prime Ministers Research Fellowship (PMRF), PMRF-ID:
1903283, Government of India.


\end{document}